\documentclass[10pt,twocolumn]{IEEEtran}
\ifCLASSINFOpdf
\else
\fi

\usepackage{amssymb}
\usepackage{amsmath}
\usepackage{amsthm}
\usepackage{graphicx}
\usepackage{algorithm}
\usepackage{algorithmic}
\usepackage{booktabs}
\usepackage{cite}
\usepackage{bm}
\usepackage{float}
\usepackage{multirow}
\usepackage{color}
\usepackage{subfigure}
\usepackage{caption}
\usepackage{epstopdf}
\usepackage{hyperref}
\usepackage{fixltx2e}
\usepackage{longtable}
\usepackage{diagbox}
\usepackage{changepage}
\usepackage{comment}
\usepackage{cases}
\usepackage[short]{newoptidef}
\usepackage{makecell}
\usepackage{mathabx}
\theoremstyle{definition}

\newtheorem{mechanism}{Mechanism}
\newtheorem{definition}{Definition}

\theoremstyle{remark}

\theoremstyle{plain}
\newtheorem{theorem}{Theorem}

\newtheorem{proposition}{Proposition}

\newcommand{\tmop}[1]{\ensuremath{\operatorname{#1}}}

\begin{document}
		%
		\title{EPViSA: Efficient Auction Design for Real-time Physical-Virtual Synchronization in the Metaverse}
		\author{Minrui Xu, Dusit Niyato, \emph{Fellow, IEEE}, Benjamin Wright, Hongliang Zhang, Jiawen Kang, Zehui Xiong,\\ Shiwen Mao, \emph{Fellow, IEEE}, and Zhu Han, \emph{Fellow, IEEE}
		\thanks{Minrui~Xu and Dusit~Niyato are with the School of Computer Science and Engineering, Nanyang Technological University, Singapore 308232, Singapore (e-mail: minrui001@e.ntu.edu.sg; dniyato@ntu.edu.sg).}
        \thanks{Benjamin Wright is with the Department of Mathematics, Massachusetts Institute of Technology, Cambridge, MA 02139 USA (e-mail: bpwright@mit.edu).}
        \thanks{Hongliang~Zhang is with the Department of Electrical and Computer Engineering, Princeton University, Princeton, NJ 08544 USA (e-mail: hongliang.zhang92@gmail.com).}
		\thanks{Jiawen~Kang is with the School of Automation, Guangdong University of Technology, China (e-mail: kavinkang@gdut.edu.cn).}
		\thanks{Zehui~Xiong is with the Pillar of Information Systems Technology and Design, Singapore University of Technology and Design, Singapore 487372, Singapore (e-mail: zehui\_xiong@sutd.edu.sg).}
		\thanks{Shiwen~Mao is with the Department of Electrical and Computer Engineering, Auburn University, Auburn, AL 36849-5201 USA (email: smao@ieee.org).}
		\thanks{Zhu~Han is with the Department of Electrical and Computer Engineering, University of Houston, Houston, TX 77004 USA, and also with the Department of Computer Science and Engineering, Kyung Hee University, Seoul 446-701, South Korea (e-mail: zhan2@uh.edu).}
	}
		\maketitle
		
		\begin{abstract}
			Metaverse can obscure the boundary between the physical and virtual worlds. Specifically, for the Metaverse in vehicular networks, i.e., the vehicular Metaverse, vehicles are no longer isolated physical spaces but interfaces that extend the virtual worlds to the physical world. Accessing the Metaverse via autonomous vehicles (AVs), drivers and passengers can immerse in and interact with 3D virtual objects overlaying views of streets on head-up displays (HUD) via augmented reality (AR). The seamless, immersive, and interactive experience rather relies on real-time multi-dimensional data synchronization between physical entities, i.e., AVs, and virtual entities, i.e., Metaverse billboard providers (MBPs). However, mechanisms to allocate and match synchronizing AV and MBP pairs to roadside units (RSUs) in a synchronization service market, which consists of the physical and virtual submarkets, are vulnerable to adverse selection. In this paper, we propose an enhanced second-score auction-based mechanism, named EPViSA, to allocate physical and virtual entities in the synchronization service market of the vehicular Metaverse. The EPViSA mechanism can determine synchronizing AV and MBP pairs simultaneously while protecting participants from adverse selection and thus achieving high total social welfare. We propose a synchronization scoring rule to eliminate the external effects from the virtual submarkets. Then, a price scaling factor is introduced to enhance the allocation of synchronizing virtual entities in the virtual submarkets. Finally, rigorous analysis and extensive experimental results demonstrate that EPViSA can achieve at least 96\% of the social welfare compared to the omniscient benchmark while ensuring strategy-proof and adverse selection free through a simulation testbed.

		\end{abstract}

		%
		\IEEEpeerreviewmaketitle

		\section{Introduction}

		\IEEEPARstart{M}{etaverse}, the immersive 3D Internet, is regarded as an advanced stage and the long-term vision of digital transformation~\cite{metaverse,lee2021all}. As an exceptional multi-dimensional and multi-sensory communication medium~\cite{mann2018all}, the Metaverse overcomes the tyranny of distance by enabling participants in different physical locations to be telepresent and interact with each other in a shared 3D virtual world~\cite{duan2021metaverse}. The synergy between the Metaverse and autonomous transportation systems, i.e., the vehicular Metaverse~\cite{jiang2021reliable}, allows drivers and passengers to access the virtual world of the Metaverse in vehicles and offers a futuristic way of traveling. The on-road immersive experience of the vehicular Metaverse relies on the persistent and bidirectional synchronization in real-time between entities in physical transportation systems and the virtual world~\cite{han2022dynamic}.
		
		The vehicular Metaverse is expected to lead an evolution of intelligent transportation systems~\cite{xiang2022comparative, zhan2021survey, zhan2020learning, le2022survey}, where ubiquitous physical and virtual entities, e.g., roadside units (RSUs), autonomous vehicles (AVs), and Metaverse billboard providers (MBPs), are interconnected. The physical-virtual synchronization system in the vehicular Metaverse is maintained from two directions. First, during the physical-to-virtual (P2V) synchronization, AVs update their digital twins (DTs) in the virtual worlds continuously by sensing surrounding environments (e.g., landmarks and landscapes) with internal and external sensors~\cite{zhang2022toward} as well as executing their DT tasks~\cite{schwarz2022role, wang2020architectural, mozaffari2020deep}, including real-time vehicular status, historical sensing data, and bio-data of passengers, in the virtual world. Then, AVs share the \textit{preference caches}\footnote{Preference caches store the intersection of AVs' historical sensing data and passengers' context-aware preferences and needs during the journey.} in the DTs to MBPs via RSUs. Finally, during the virtual-to-physical (V2P) synchronization, MBPs push personalized augmented reality (AR)-based recommendations and ads~\cite{luettin2019future} to AVs and display on head-up displays (HUDs)~\cite{wang2020augmented}. This bidirectional synchronization can provide the utmost safety and immersive experience to drivers and passengers. 
		However, tremendous communication and computing resources might be consumed to facilitate real-time physical-virtual synchronization in the vehicular Metaverse~\cite{xu2022full}, including maintaining DT models of AVs and delivering virtual billboards of MBPs.
		Fortunately, RSUs, equipped with high-capacity communication and computing resources, can provide synchronization services to AVs and MBPs~\cite{dai2021asynchronous}. Therefore, synchronizing AVs and MBPs can improve their immersive experience by offloading resource-intensive DT and AR tasks to RSUs. Nevertheless, RSUs are expected to require compensation since maintaining synchronization services will consume RSUs' limited resources. Hence, a real-time synchronization service market should be designed to motivate RSUs with monetary incentives to share their resources by allocating and charging the synchronizing AV and MBP pairs in the vehicular Metaverse.
		\begin{figure*}
			\centering
			\includegraphics[width=1\linewidth]{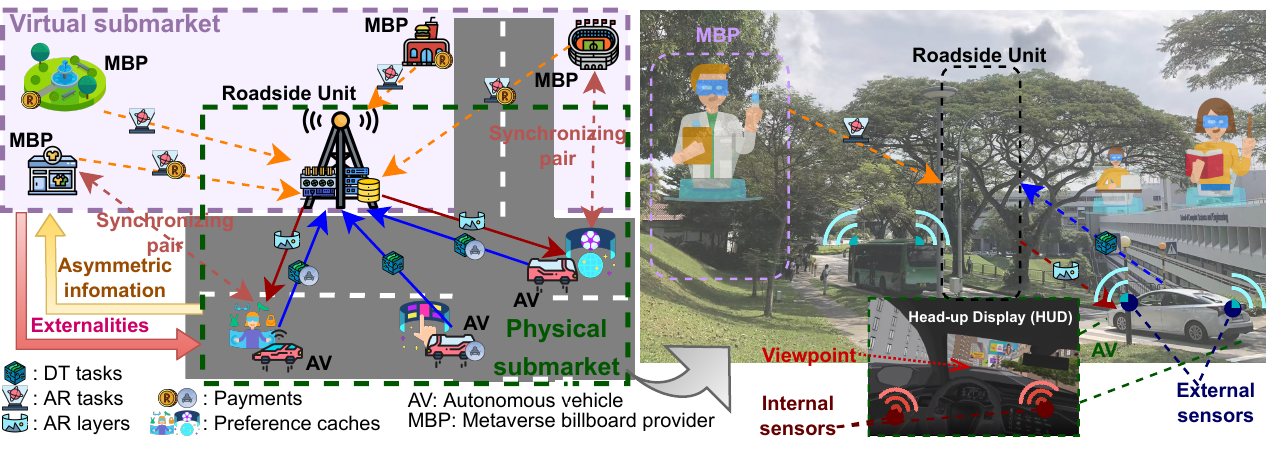}
			\caption{The physical-virtual synchronization system in the vehicular Metaverse.}
			\label{fig:system}
		\end{figure*}
		In literature, many efforts have been devoted to designing auction-based mechanisms in the synchronization service market since auctions can effectively deal with private information while not introducing an excessive latency~\cite{tang2017momd}. However, existing solutions have limitations and are unsuitable for real-time physical-virtual synchronization~\cite{ismail2022semantic} as they allocate synchronizing physical or virtual entities separately. As a result, they introduce unnecessary delays~\cite{xu2022lsync} during the physical-virtual synchronization of AVs and MBPs via RSUs that degrades the seamless experience perceived by passengers being in the Metaverse~\cite{chen2018virtual}. 
		Moreover, they primarily assume that there are no external effects (or externalities)~\cite{zhang2021privacy} from the synchronizing virtual entities during the allocation of physical entities in the P2V synchronization~\cite{du2021optimal}. However, the synchronizing virtual entities, e.g., MBPs, may introduce externalities to the allocation of AVs as MBPs' recommendations and ads are diverse in match qualities and render delays. These unpredictable qualities and delays introduce externalities and expose AVs to \textit{adverse selection}\footnote{Adverse selection~\cite{akerlof1978market} refers to a market situation that participants with asymmetric information are only willing to pay at the average market price. This might introduce an inefficient allocation outcome of the market~\cite{arnosti2016adverse}.} that AVs demand RSUs to set a prefixed threshold of synchronization delays~\cite{hui2022collaboration} to minimize externalities.
		Furthermore, the existing solutions mostly assume that virtual entities have a symmetric observation of the physical world~\cite{xu2022wireless}. In contrast, the match qualities of recommendations and ads might not be immediately measured by every MBP during the V2P synchronization. Under such an asymmetric information regime, MBPs without immediate feedback cannot adjust their bids accordingly and thus are exposed to adverse selection. In this way, they might lose the opportunities valued to them most.
		Consequently, the existing solutions may expose AVs and MBPs to adverse selection, and the allocation efficiency in the synchronization market may deteriorate.

		In this paper, to address the problem of adverse selection and improve allocation efficiency, we propose an \textbf{E}nhanced second-score auction-based \textbf{P}hysical-\textbf{Vi}rtual \textbf{S}ynchroniz\textbf{A}tion mechanism (EPViSA) to allocate and match physical and virtual entities simultaneously in the synchronization service market of the vehicular Metaverse. The EPViSA mechanism determines the set of synchronizing AVs and MBPs in the physical and virtual submarkets, respectively. 
		In the physical submarket for the P2V synchronization, to address the adverse selection problem of AVs, we propose a synchronization scoring rule that incorporates information from MBPs in the virtual submarket. The efficient synchronization scoring rule allows the RSUs to adjust processing deadlines flexibly according to the requirements and preferences of AVs. Additionally, in the virtual submarket for the V2P synchronization, the enhanced allocation and the pricing rules with a price scaling factor are proposed to allocate and price the synchronizing MBPs with high-quality recommendations and ads. This way, the problem of asymmetric information among MBPs is alleviated, and thus the surplus in the virtual submarket increases.
		To demonstrate the allocation efficiency of the EPViSA mechanism, we first provide an omniscient mechanism and a second-price auction-based mechanism as benchmarks~\cite{arnosti2016adverse}. We analyze and prove a lower bound on the allocation efficiency of the proposed mechanisms compared to the omniscient benchmark. Moreover, we show that the proposed EPViSA mechanism can achieve at least 96\% of the social welfare in the synchronization market compared with the omniscient benchmark. Furthermore, a simulation testbed of the vehicular Metaverse is implemented to collect real-world data samples, e.g., opinion scores and eye fixation duration, during users watch Metaverse billboards.

		Our main contributions can be summarized as follows:
		\begin{itemize}
			\item We propose a new physical-virtual synchronization system for the vehicular Metaverse, where RSUs provide real-time synchronization services to help AVs synchronize with the virtual worlds by executing DT tasks and help MBPs synchronize with the physical world by rendering AR recommendations and ads.

			\item In the proposed system, we design a real-time synchronization service market where service providers, i.e., RSUs, allocate and match the synchronizing pairs of AVs and MBPs, whose surpluses are positively correlated in the physical and the virtual submarkets.
			\item To resolve the adverse selection problem in existing synchronization mechanisms, we propose the EPViSA mechanism to match and allocate AVs and MBPs for RSUs simultaneously. In the physical world, the synchronization scoring rule is proposed for synchronizing AV allocation to eliminate the external effect from the virtual world. In the virtual world, the pricing factor is adopted to mitigate the asymmetric information in matching qualities of MBPs' recommendations and ads.
			\item To evaluate the performance of the proposed mechanism, extensive analysis and experiments are conducted. We implement a simulation testbed of the vehicular Metaverse to collect real-world data samples for validating the proposed mechanisms. The analysis and experimental results show that the EPViSA mechanism can achieve at least 96\% of the allocation efficiency in the synchronization market while guaranteeing fully strategy-proof and adverse-selection free.
		\end{itemize}
		
		The rest of this paper is organized as follows. In Section~\ref{sec:related}, we provide a review of the related work. In Section~\ref{sec:system}, we describe the synchronization system and market. In Section~\ref{sec:mechanism}, we provide the implemented details of the proposed synchronization mechanism. In Section~\ref{sec:enhanced}, we propose the enhanced synchronization mechanism. We illustrate the simulation experiments in Section~\ref{sec:results}, and finally conclude in Section~\ref{sec:conclusions}.
		
		\section{Related Work}\label{sec:related}
		\subsection{Metaverse}
		Since the term ``Metaverse" was introduced by Neal Stephenson in his \textit{Snow Crash} 30 years ago, it has raised a continuous exploration of human beings in its potential impacts on the economy, productivity, and society~\cite{wang2022survey}. The Metaverse is expected to provide extended reality and extended intelligence in all daily activities, including entertainment, work, and socializing. To put the concept of the Metaverse into practice, Duan~\textit{et al.}~\cite{duan2021metaverse} implemented a Metaverse prototype of a university campus for virtual education~\cite{ng2022unified} based on blockchain~\cite{hong2022scaling, hong2022cycle} and immersive human-computer interfaces. To clarify the demands and questions in the current development of the Metaverse, they proposed a three-layer architecture for the Metaverse, which consists of the ecosystem, interaction, and infrastructure. Specifically, in the interaction layer of their architecture, DT and immersive user experience are indispensable to support the synchronization among entities in the physical and virtual worlds.

		\subsection{Physical-Virtual Synchronization in Vehicular Metaverse}
		The immersive experience of passengers in the vehicular Metaverse relies on the real-time synchronization among physical entities, e.g., AVs~\cite{hossain2022ai, hu2021uav, ni2019toward}, and virtual entities, e.g., MBPs~\cite{xu2022secure,zhang2020machine,zhang2022intelligent}. Aiming to relieve the computation burden on AVs during synchronization from the physical world to the virtual world, Jiang \textit{et al.}~\cite{jiang2021reliable} investigated the application of coded distributed computing (CDC) in the vehicular Metaverse with a hierarchical decision-making framework. Hui \textit{et al.}~\cite{hui2022collaboration} demonstrated that maintaining real-time DT models for interconnected AVs can improve the efficiency of intelligent transportation systems~\cite{liu2021smart} by making collaborative decisions in the virtual world. Additionally, to synchronize the virtual world with physical vehicles, Wang \textit{et al.}~\cite{wang2020augmented} designed Metaverse billboards by leveraging AR as the interface that overlays the recommendations and ads on the field-of-view of passengers through the windshield. Through real-time bidirectional synchronization between physical and virtual worlds, i.e., maintaining DT models, and from virtual to physical worlds, i.e., delivering AR recommendations and ads, the vehicular Metaverse is regarded as an enabler of future transportation systems.

		\subsection{Mechanism Design for Metaverse Synchronization Market}
		
		Several existing works have been done to design mechanisms for Metaverse service markets, with the aim of providing allocation and pricing mechanisms for physical and virtual services in the Metaverse. For example, Ismail \textit{et al.}~\cite{ismail2022semantic} propose a Vickrey–Clarke–Groves (VCG) auction-based mechanism for IoT devices to sell their semantic information to virtual service providers in the Metaverse that is responsible for constructing and maintaining the Metaverse. In addition, Xu \textit{et al.}~\cite{xu2022wireless} propose a double Dutch auction-based VR services allocation and pricing mechanism for physical VR service providers to provide their services to VR users.
		However, their current works generally focus on a unilateral market for physical or virtual services. In our work, we propose a family of auction-based physical-virtual synchronization mechanisms to simultaneously allocate and match the synchronizing AV and MBP pairs in the physical and virtual submarkets. Different from existing synchronization mechanisms, the proposed EPViSA mechanism can achieve real-time and efficient allocation and match for synchronizing AV and MBP pairs while ensuring several promising economic properties, such as strategy-proofness and adverse selection free.
		\begin{table}[t]
			\setlength{\abovecaptionskip}{5pt}
			\setlength{\belowcaptionskip}{5pt}
			\renewcommand{\arraystretch}{1.3}
			\caption{List of notations.}
			\label{table1}
			\centering
			\begin{tabular}{c || p{6.2cm}}
				\toprule
				Notation & Description     \\
				\hline
				$\mathcal{I}, \mathcal{J}, \mathcal{K}$         & Sets of AVs, RSUs, and MBPs.           \\
				$f^C_j, f^G_j$ &The CPU and GPU frequencies of RSU $j$.\\
				$B_j,B^u_j,B^d_j$ & The total, uplink, and downlink bandwidths of RSU $j$.\\
				$\nu_i$ & The valuation of the DT task of AV $i$.\\
				$m_{i,k}$ & The match quality of MBP $k$ on AV $i$.\\
				$Q_{i,k}$ & The valuation of ads of MBP $k$ on AV $i$.\\
				$C_i$ & The size of preference caches of AV $i$.\\
				$h_{i,k}$ & The number of matched preference caches of AV $i$ by MBP $k$.\\
				$P_j, p_i$ & The transmit power of RSU $j$ and AV $i$.\\
				$g_{i,j}$ & The channel gain between RSU $j$ and AV $i$.\\
				$R^u_{i,j}, R^d_{i,j}$ & The uplink and downlink transmission rates between AV $i$ and RSU $j$.\\
				$s_i^\emph{DT}, e_i^\emph{DT}, \eta_i$ & The size of data, the required CPU cycles per unit data, and the deadline to execute AV $i$'s DT task.\\
				$\eta_i$ & The deadline for accomplishing the DT task of AV $i$.\\
				$t_{i,j}^\emph{DT}, d_{i,j}^\emph{DT}$ & The transmission and computation delays to execute the DT task of AV $i$ at RSU $j$.\\
				$s_k^\emph{AR},e_k^\emph{AR}$ & The data size of the AR layer and the required GPU cycles per unit data of MBP $k$.\\
				$t_{i,j,k}^\emph{AR}, e_{i,j,k}^\emph{AR}$ & The transmission and computation delays of AR recommendations and ads of MBP $k$ on RSU $j$ to AV $i$.\\
				$T^\emph{total}$ & The synchronization delay.\\
				$W^\emph{DT}$ & The surplus of AVs.\\ 
				$W^\emph{AR}_B, W^\emph{AR}_P$ & The surplus of brand MBPs and performance MBPs.\\ 
				$b^\emph{DT}, b^\emph{AR}$ & The submitted bids of AVs and MBPs.\\
				$z^\emph{DT}, z^\emph{AR}$ & The allocation probabilities of AVs and MBPs.\\ 
				$p^\emph{DT}, p^\emph{AR}$ & The payments of AVs and MBPs.\\ 
				$\Phi^\emph{syn}$ & The synchronization scores of AVs.\\
				\bottomrule
			\end{tabular}
		\end{table}
		
		\section{The Synchronization System and Market}~\label{sec:system}
		In this section, we first present an overview of the synchronization system in the vehicular Metaverse and discuss the problem of adverse selection in the synchronization market in Subsections~\ref{pvsyn} and \ref{overview}. Then, we describe the network model in Subsection~\ref{network}, the DT task model in Subsection~\ref{dttask}, and the AR recommendation and ads model in Subsection~\ref{arads}. Finally, we provide the problem formulation in Subsection~\ref{problem} for the allocation and matching of synchronizing AR and MBP pairs.
		
		\subsection{Physical-Virtual Synchronization in Vehicular Metaverse}\label{pvsyn}
		
		The vehicular Metaverse will provide a safe and immersive experience to passengers by synchronizing the physical and the virtual worlds in a real-time manner. During the P2V synchronization, each AV generates and updates its DT model via RSUs and maintains the preference caches based on the AV's historical sensing data and passengers' preferences and needs during the journey. During the V2P synchronization, MBPs deliver personalized AR recommendations and ads to AVs based on the preference caches.
		
		\subsection{An Overview of the Physical-Virtual Synchronization System}\label{overview}
		As shown in Fig. \ref{fig:system}, we consider a real-time physical-virtual synchronization system for the vehicular Metaverse that consists of AVs, RSUs, and MBPs. The key notations used in this work are summarized in Table \ref{table1}.
		\subsubsection{Autonomous Vehicles (AVs)} The set of AVs is denoted as $\mathcal{I}=\{1,\dots,i,\dots,I\}$. In the synchronization system, physical AVs synchronize with the virtual world by continuously generating and executing updates of their DT models. Continuously generated DT tasks arriving at AV $i$ can be represented by $DT_i = \textless s_i^\emph{DT}, e_i^\emph{DT}, \eta_i \textgreater$, where $s_i^\emph{DT}$ is the size of DT data, $e_i^\emph{DT}$ represents the required CPU cycles per unit data, and $\eta_i$ denotes the deadline for accomplishing the task. Each AV $i\in\mathcal{I}$ has the private valuation $\nu_i$ for executing its DT task $DT_i$ that are drawn from the probability distribution $G_i$. The valuations of DT tasks can be interpreted as the capture attributes of AVs, e.g., the urgency to synchronize with DT models~\cite{hui2022collaboration}, which may be different and varying for each AV during its journey.
		\subsubsection{Roadside Units (RSUs)} The set of RSUs is denoted as $\mathcal{J} = \{1,\dots,j,\dots,J\}$. Let $B_j$ denote the total bandwidth, and $P_j$ denote the transmit power allocated by RSU $j$. The bandwidth of each RSU is divided into the downlink bandwidth $B_j^d$ and the uplink bandwidth $B_j^u$, where $B_j^u + B_j^d = B_j$. To provide synchronization services, e.g., executing DT tasks and rendering AR layers, each RSU $j$ is equipped with computing capacities, including the CPU frequency $f^{C}_j$ and the GPU frequency $f^{G}_j$.
		\subsubsection{Metaverse Billboard Providers (MBPs)} The set of MBPs is denoted as $\mathcal{K}=\{0, 1,\dots, k, \dots, K\}$. Similar to the Internet display advertising~\cite{arnosti2016adverse}, we consider two types of MBPs in the system, i.e., brand MBPs
		and performance MBPs.
		The performance MBPs $1,\dots,K$ provide recommendations and performance ads that are designed to elicit real-time feedback from passengers, such as navigation to a shop along the street where a sale may occur.
		The brand MBP $0$ delivers recommendations and brand ads that are designed to raise passengers' awareness of a brand or product, such as infotainment of a sales event and vehicle maintenance for an upcoming opportunity.
		The valuation of personalized AR ads is $Q_{i,k}$ for each synchronizing pair of AV $i$ and MBP $k$, which is the product of the common valuation $\nu_{i}$ of AV $i$ and the match quality $m_{i,k}$, i.e., $Q_{i,k} = \nu_{i}m_{i,k}$. Straightforwardly, common valuations for every MBP $k$ are gained from the provisioning of general recommendations for the synchronizing AV $i$, which can be represented by the AV $i$'s private valuation $\nu_i$. Additionally, the match quality $m_{i,k}$ of MBP $k$ is determined by the amount of personalized information,
		This way, the valuations of AVs and MBPs in synchronizing pairs are positively correlated. Finally, let $Q_{\iota,(l)}$ and $m_{\iota,(l)}$ represent the $l$ highest valuation and match quality for the synchronizing AV $\iota$, respectively.

		\subsubsection{Preference Caches}
		Based on historical internal and external sensing data, preferences, and needs, passengers in AVs maintain preferences of different types of landmarks and landscapes that are stored in preference caches\footnote{The preference caches are uploaded along with AVs' DT models during P2V synchronization.}. Here, let $C_i$ denote the size of the preference caches of each AV $i$. The numbers of matched preference caches $H \in \mathbb{R}_+^{I\times (K+1)}$ of matched AVs and MBPs can be expressed as 
		\begin{equation}
			H =  \begin{bmatrix}
				h_{1,0} &  h_{2,0} & \cdots & h_{I,0} \\
				h_{1,1}& h_{2,1} & \cdots & h_{I,1} \\
				\vdots & \vdots & \ddots & \vdots \\
				h_{1,K}& h_{2,K} & \cdots & h_{I,K}
			\end{bmatrix},
		\end{equation}
		where $0\leq h_{i,k} \leq C_i, \forall i\in \mathcal{I}, k\in \mathcal{K}$ denotes the number of matched preference caches between AV $i$ for MBP $k$.
		Given the total number of MBPs $K+1$, the match qualities $m_{i,k}$ are drawn independently from a set of distributions $m_{i,k} = h_{i,k} \sim F_{i,k}$. 
		
		\subsubsection{Location-based Personalized AR ads}
		
		Based on preference caches of AVs, MBPs distribute personalized brand and performance ads to AVs, each of which consists of one basic layer\footnote{The basic layers consist of some general location-based recommendations and ads, such as weather and real-time traffic monitoring.} and multiple enhancement layers\footnote{The enhancement layers consists of virtual objects of personalized context-aware recommendations and ads based on the matched preference caches, such as navigation and vehicular maintenance notification.}. For each AR ads layer of MBP $k$, the AR rendering task can be represented by $AR_k = \textless s_k^\emph{AR}, e_k^\emph{AR} \textgreater$~\cite{ren2020edge}, where $s_k^\emph{AR}$ is the data size of each AR layer and $e_k^\emph{AR}$ is the required GPU cycles per unit data for rendering. Therefore, given the total number of MBPs $K+1$, the match quality $m_{i,k}$ are drawn independently from a set of distributions $m_{i,k} = h_{i,k} \sim F_{i,k}$. To explain further, given the synchronizing AV $\iota$, the performance MBPs $k = 1,\dots, K$ can measure the match qualities $m_{\iota, k}$ of their recommendations and performance ads. However, the brand MBP $0$ that provides recommendations and brand ads to the synchronizing AV $\iota$ cannot measure its match quality $m_{\iota, 0}$ immediately. Therefore, asymmetric information exists among MBPs that might result in adverse selection.
		
		\subsection{Adverse Selection in the Synchronization Market}\label{adverse}
		
		Adverse selection~\cite{akerlof1978market} refers to a market situation where participants with asymmetric information are only willing to pay at the average market price. This might introduce an inefficient allocation and matching outcome of the synchronization market. In the synchronization market, the physical and virtual entities have positively correlated surpluses for synchronization services. In other words, the surplus of AVs in the physical submarket can affect the surplus of MBPs in the virtual submarket by affecting the valuation of the common valuation of AR ads. Therefore, such correlation introduces externalities and asymmetric information for allocating synchronizing physical and virtual entities in the synchronization service market, as follows.
		\begin{itemize}
			\item Externalities: The externalities are introduced to the physical submarket from the virtual submarket. The personalized AR recommendations and ads of MBPs have different match qualities for different physical AVs. However, during the allocation of synchronizing AV in the physical market, the synchronizing MBP is unknown for participants in the physical submarket, which might affect the total processing delay in the physical submarket. Therefore, AVs in the physical submarket prefer to demand RSU to set a prefixed threshold of synchronization delay before the allocation of the synchronizing AVs.
			\item Asymmetric Information: There exists asymmetric information among MBPs for their personalized recommendations and ads. The performance ads (e.g., navigation to a shop) can induce immediate responses from users. In contrast, brand ads (e.g., infotainment and vehicle maintenance notifications) cannot be measured by brand MBPs immediately.
		\end{itemize}
		
		\subsection{Network Model}\label{network}
		During the real-time physical-virtual synchronization, uplink and downlink transmissions are utilized for updating DTs and delivering AR Metaverse billboards~\cite{zhang2021energy}, respectively. Let $g_{i,j}$ denote the channel gain between AV $i$ and RSU $j$, the downlink transmission rate $R^d_{i,j}$ can be calculated as~\cite{chen2018virtual} $
		R^{d}_{i,j} = B_j^d \log(1+\frac{g_{i,j}P_j}{\sigma_i^2}),$
		where $\sigma^2_i$ is the additive white Gaussian noise at AV $i$.
		Moreover, let $p_i$ denote the transmit power of AV $i$, the uplink transmission rate $R^u_{i,j}$ can be calculated as $
		R^{u}_{i,j} = B_j^u \log(1+\frac{g_{i,j}p_i}{\sigma_j^2}),
		$
		where $\sigma^2_j$ is the additive white Gaussian noise at RSU $j$.

		\subsection{DT Task Model}\label{dttask}
		
		To synchronize with the vehicular Metaverse, physical entities, i.e., AVs, generate and offload DT synchronizing demands, i.e., DT model updates, to RSUs for real-time execution. Therefore, we consider the demands as tasks that are required to be accomplished by RSUs.  The transmission delay for AV $i$ to upload its DT task $DT_i$ to RSU $j$ can be calculated as~\cite{hui2022collaboration} $
		t_{i,j}^\emph{DT} = \frac{s_i^\emph{DT}}{R_{i,j}^{u}}.
		$ The computation delay in processing the DT task $DT_i$ of AV $i$ for RSU $j$ can be calculated as $
		d_{i,j}^\emph{DT} = \frac{s_i^\emph{DT}e_i^\emph{DT}}{f^{C}_{j}}.$
		In the proposed system, without loss of generality, we consider that each RSU has the capability to accomplish both computing and transmission requirements of DT tasks, i.e., $ t_{i,j}^\emph{DT} + d_{i,j}^\emph{DT} \leq \eta_i, \forall i\in \mathcal{I}, j\in\mathcal{J}$. With the use of available communication and computing resources, RSUs can provide AR rendering services for MBPs. This way, MBPs can send their AR recommendations and ads to AVs, i.e., synchronizing from the virtual world to the physical world. 
		
		\subsection{AR Recommendation and Ads Model}\label{arads}
		In the synchronization system, MBPs synchronize personalized AR recommendations and ads that provide an on-road immersive experience to passengers. In the vehicular Metaverse, passengers inside vehicles can observe not only the scenery outside but also AR recommendations and ads overlaying real-world landmarks and landscapes through HUDs.
		The transmission delay in rendering and transmitting the AR recommendations and ads $AR_k$ to AV $i$ from RSU $j$ can be calculated as 
		\begin{equation}
			t_{i,j,k}^\emph{AR} = \frac{(h_{i,k}+1)s_k^\emph{AR}}{R_{i,j}^{d}}.
		\end{equation}
		The computation delay in rendering the AR recommendations and ads $AR_k$ can be calculated as
		\begin{equation}
			d_{i,j,k}^\emph{AR} = \frac{(h_{i,k}+1)s_k^\emph{AR}}{f^C_{j}} + \frac{(h_{i,k}+1)s_k^\emph{AR}e_k^\emph{AR}}{f^G_{j}},
		\end{equation}
		which consists of the input delay for the CPU and the rendering delay for the GPU of RSU $j$.
		
		In the synchronization system, RSUs can use their available computation and communication resources to provide real-time physical-virtual synchronization services for AVs and MBPs. However, the total synchronization delay cannot exceed the required deadline of AV $i$. Let $z_{i,j}^\emph{DT}$ be the allocation variable that AV $i$ is allocated to RSU $j$ and $z_{i,j,k}^\emph{AR}$ be the allocation variable that MBP $k$ is allocated by RSU $j$ to match AV $i$. The total synchronization delay $T^\emph{total}_{i,j,k}$ required by RSU $j$ to process both the DT task of AV $i$ and the AR rendering the task of MBP $k$ should be less than the required deadline, which can be expressed as
		\begin{equation}
			T^\emph{total}_{i,j,k} = z_{i,j}^\emph{DT}\cdot(t_{i,j}^\emph{DT} + d_{i,j}^\emph{DT}) + z_{i,j,k}^\emph{AR}\cdot(t_{i,j,k}^\emph{AR} +d_{i,j,k}^\emph{AR}) \leq \eta_i,
		\end{equation}
		$ \forall i \in \mathcal{I}, j\in\mathcal{J}, k\in\mathcal{K}$. 
		Moreover, when RSUs process DT tasks for AVs, they also display AR recommendations and ads of MBPs to AVs. Therefore, $T^\emph{total}_{i,j,k}$ also represents the expected displaying duration of the AR recommendations and ads of MBP $k$ for AV $i$ with the aid of RSU $j$.
		\subsection{Problem Formulation}\label{problem}
		In the synchronization system, there are two concurrent non-cooperative games for AVs and MBPs competing for the limited synchronization services of RSUs. To motivate RSUs to contribute their resources to the synchronization system, we design a real-time synchronizing market where sellers (RSUs) can earn profits from provisioning synchronization services, and physical and virtual bidders (AVs and VRSs) are competing for synchronization services in the physical and the virtual submarket, respectively.
		We consider bidders (AVs and MBPs) in the synchronization market to be risk neutral, and their surpluses are positively correlated with each other based on the revelation principle~\cite{arnosti2016adverse}.
		Based on this, we consider a mechanism to be a mapping from the private information $\nu=(\nu_1, \dots, \nu_I)$ and $Q = (Q_{1,0}, \dots,Q_{I,K})$ to allocation probabilities $z^\emph{DT} = (z^\emph{DT}_1,\dots, z^\emph{DT}_I)$ and $z^\emph{AR} = (z^\emph{AR}_0, \dots, z^\emph{AR}_K)$ as well as payments $p^\emph{DT}=(p^\emph{DT}_1, \dots, p^\emph{DT}_I)$ and $p^\emph{AR}=(p^\emph{AR}_0, \dots, p^\emph{AR}_K)$. 
		
		The expected surplus for each RSU to accomplish DT tasks for AV $i\in\mathcal{I}$ is defined by
		\begin{equation}
			W^\emph{DT}(z^\emph{DT}) = \mathbb{E}\left[\sum_{i=1}^{I} \nu_i z^\emph{DT}_{i,j}(\nu)\right].
		\end{equation}
		Then, the expected surplus from the AR ads position awarded to brand MBP $0$ is given by $W_\emph{B}^\emph{AR} = \mathbb{E}[Q_{i,0}z^\emph{AR}_{i,j, 0}(Q_i)]$. Similar to the surplus of the brand MBP, the expected surplus from the AR ads position assigned to performance MBPs can be defined as $W_\emph{P}^\emph{AR}(z^\emph{AR}) = \mathbb{E}[\sum_{k=1}^{K}Q_{i,k}z^\emph{AR}_{i,j,k}(Q_i)]$. Therefore, for the matching AV $\iota$, the unit surplus obtained from the AR ads position can be calculated as the weighted sum of these two terms, i.e.,
		\begin{equation}
			W^\emph{AR}(z^\emph{AR}) = \gamma W_\emph{B}^\emph{AR}(z^\emph{AR}) + W_\emph{P}^\emph{AR}(z^\emph{AR}),
		\end{equation}
		where $\gamma$ is the relative bargaining power of brand MBP $0$.
		
		As for the displaying costs of AR recommendations and ads, we consider the cost-per-time model, i.e., MBPs pay according to the total displaying duration for their recommendations and ads. Therefore, the total surplus of MBPs is the product of the expected displaying duration and the total surplus of the virtual submarket, i.e., $T^\emph{total}_{\iota,j,k} \cdot W^\emph{AR}(z^\emph{AR}_\iota)$.

		Our objective is to maximize the total social welfare of AVs and MBPs with positively correlated surpluses by utilizing the limited synchronization services of RSUs. Therefore, the optimization problem for real-time mechanism $\mathcal{M}=(z^\emph{DT}, z^\emph{AR})$ in the synchronization service market can be formulated as
		\begin{maxi!}|s|[2]<b>
			{\mathcal{M}}{W^\emph{DT}+ T^\emph{total}_{i,j,k} \cdot \big(\gamma W_B^\emph{AR}+W_P^\emph{AR}\big)}{}{}
			\addConstraint{T^\emph{total}_{i,j,k}}{\leq \eta_i, \quad }{\forall i\in \mathcal{I}, j\in\mathcal{J}, k\in\mathcal{K}\label{con1}}{}{}
			\addConstraint{h_{i,k}}{\leq C_i, \quad}{\forall i \in \mathcal{I}, k\in \mathcal{K}\label{con2}}{}{}
			\addConstraint{\sum_{i=1}^{I}z_{i,j}^\emph{DT}}{\leq 1,\quad}{\forall j \in \mathcal{J}\label{con3}}{}{}
			\addConstraint{\sum_{k=0}^{K}z_{i,j,k}^\emph{AR}}{\leq 1,\quad}{\forall i \in \mathcal{I}, j \in \mathcal{J}\label{con4}}{}{}
			\addConstraint{z_{i,j}^\emph{DT}}{\in \{0,1\},\quad}{\forall i \in \mathcal{I}, j \in \mathcal{J}\label{con5}}{}{}
			\addConstraint{z_{i,j,k}^\emph{AR}}{\in \{0,1\},\quad}{\forall i \in \mathcal{I}, j \in \mathcal{J}, k \in \mathcal{K}\label{con6}}{}{}.
		\end{maxi!}
		The reliability constraint \eqref{con1} guarantees that the total delay of executing DT and AR tasks is lower than the deadline requirement of the synchronizing AV, and the preference cache constraint \eqref{con2} ensures that the number of match qualities is less than the size of preference caches. The allocation constraints \eqref{con3},\eqref{con4}, \eqref{con5}, and \eqref{con6} ensure each AV and MBP can only allocate and be allocated by one RSU.
		To evaluate this allocation optimization problem of the real-time synchronization market, an omniscient mechanism $\mathcal{M}^{*}$~\cite{arnosti2016adverse} with maximum achievable total social welfare is treated as the omniscient benchmark for the following synchronization mechanism design. The allocation rule of $\mathcal{M}^{*}$ is to award the position of the ads to the performance MBP with the highest match quality whenever $m_{\iota,(1)}>\gamma\mathbb{E}[m_{\iota,0}]$ and to the brand MBP otherwise. However, this allocation rule is practically impossible to implement since it requires the deduction of match qualities from common valuations, which are invisible for the brand MBP. Although $\mathcal{M}^{*}$ provides an achievable upper bound\footnote{The omniscient mechanism can obtain social welfare at around $\nu_{\iota} + \eta_{\iota}\left[\gamma \mathbb{E}[Q_{\iota,0}]+\mathbb{E}[Q_{\iota,(1)}]\right]$, by selecting the optimal physical bidder $\iota \in \mathcal{I}$.} of efficiency and measurable losses for the following mechanism designs, its allocation rule is in general not implementable due to the fact that $\nu$ is unobservable in the virtual submarket.
		
		
		\section{The Synchronization Mechanism}\label{sec:mechanism}
		In this section, we first discuss some preliminaries for the synchronization mechanism in Subsection~\ref{preliminaries}. Then, we propose a simple second-price auction-based physical-virtual synchronization mechanism in Subsection~\ref{pvisa}, named PViSA. Finally, we provide a property analysis of the efficiency of the PViSA mechanism in Subsection~\ref{pvisaproperty}.

		\subsection{Preliminaries}\label{preliminaries}
		We design an auction-based synchronization mechanism in which each RSU simultaneously allocates and matches AVs and MBPs for physical-virtual synchronization by executing DT tasks and rendering AR layers. At each decision slot, an auctioneer, the RSU, or a proxy delegated by RSUs, initiates an auction to facilitate competition for its synchronization services. The mechanism operates in an asynchronous and decentralized manner as each RSU, or its proxy, can initiate an auction independently. For tractability and without loss of generality, we consider the case that one RSU runs an auction for the design of mechanisms.
		
		The physical-virtual synchronization market consists of two submarkets, i.e., the physical submarket for the P2V synchronization and the virtual submarket for the V2P synchronization, which operate concurrently. In the physical submarket, physical entities, i.e., AVs, submit their bids to the auctioneer for executing DT tasks on RSUs. In the virtual submarket, virtual entities, i.e., MBPs, submit their bids to the auctioneer for rendering and distributing AR recommendations and ads via RSUs. The objective of the auctioneer is to assign the synchronization services of RSUs to a synchronizing AV and MBP pair that values it most, i.e., the winning AV and MBP pair, which ensures some promising economic properties, such as strategy-proofness~\cite{krishna2009auction}. By jointly considering the limited resources of RSU, the required execution deadline of DT tasks, and the bidding prices from ARs and MBPs, the auctioneer determines the winning AV and MBP pair and the corresponding payments for the synchronization services.
		During the auction of synchronization services, the brand MBP chooses $b_0$ (in a contract~\cite{zhang2017non}) to maximize its expected utility~\cite{arnosti2016adverse}, as given by $
		U^\emph{AR}_0(b^\emph{AR}) = \mathbb{E}[Q_{\iota,0} - Q_{\iota,(1)} 1_{Q_{\iota, (1)}\leq b_0^\emph{AR}}]$, where $\iota$ is the winning AV in the physical submarket. The structure of the synchronization market is illustrated in Fig.~\ref{fig:mechanism}.
		\begin{figure}
			\centering
			\includegraphics[width=1\linewidth]{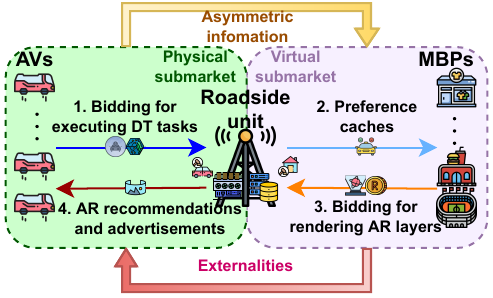}
			\caption{The structure of the physical-virtual synchronization market in the vehicular Metaverse.}
			\label{fig:mechanism}
		\end{figure}
		
		\subsection{Designing the PViSA Mechanism}\label{pvisa}
		In the second-price auction-based physical-virtual mechanism, the auctioneer adopts the second-price allocation rule and pricing rule~\cite{krishna2009auction} simultaneously in the physical submarket and the virtual submarket to allocate and match synchronizing AV and MBP pairs. The PViSA mechanism works in two phases for the physical and the virtual submarkets simultaneously, as follows.
		\begin{mechanism}[PViSA Mechanism]
			The PViSA mechanism can be characterized by the allocation rule and pricing rule in both physical and virtual submarkets as $\mathcal{M}^\emph{PViSA} = (\mathcal{Z}^\emph{PViSA}_\emph{phy}, \mathcal{P}^\emph{PViSA}_\emph{phy},\mathcal{Z}^\emph{PViSA}_\emph{vir},\mathcal{P}^\emph{PViSA}_\emph{vir})$. Before the auction of the PViSA mechanism, each RSU $j$ first announces its threshold deadline $\hat{\eta}_j$ to the auctioneer that $T^\emph{total}_{i,j,k}\leq \hat{\eta}_j \leq \eta_i, \forall i\in\mathcal{I}$.
			\begin{enumerate}
				\item Physical submarket:
				In the physical submarket, AVs submit their bids $b^\emph{DT} = (b^\emph{DT}_1, \dots, b^\emph{DT}_I)$ for the computation and communication resources of RSU $j$. The auctioneer determines the allocation rule and pricing rule for the physical submarket in PViSA as follows.
				\begin{itemize}
					\item Allocation rule: The auctioneer calculates the allocation probabilities $z^\emph{DT}_i$ for every AV $i$ by evaluating whether it submits the highest bid, i.e.,
					\begin{equation}
						\mathcal{Z}^\emph{PViSA}_\emph{phy}  : z_i^\emph{DT}(b^\emph{DT}) = 1_{\{b_i^\emph{DT} > \max \{b_{-i}^\emph{DT}\}\}}.
					\end{equation}
					\item Pricing rule: Following the pricing rule of the second-price sealed-bid auction, the auctioneer calculates the payments for each bidder as
					\begin{equation}
						\mathcal{P}^\emph{PViSA}_\emph{phy}  : p_i^\emph{DT}(b^\emph{DT}) =  z_i^\emph{DT}(b^\emph{DT}) \cdot \max  \{b^\emph{DT}_{-i}\}.
					\end{equation}
				\end{itemize}
				\item Virtual submarket: Based on the personal attributes and preference caches of the winning AV in the physical submarket, MBPs in the virtual submarket submitted their bids $b^\emph{AR} = (b^\emph{AR}_0, b^\emph{AR}_1, \dots, b^\emph{AR}_K)$ to the auctioneer. Then, the auctioneer executes the allocation rule and the pricing rule of the PViSA mechanism in the virtual submarket as follows.
				\begin{itemize}
					\item Allocation rule: The ads position is allocated to the synchronizing MBP submitted the highest bid as
					\begin{equation}
						\mathcal{Z}^\emph{PViSA}_\emph{vir} : z^\emph{AR}_k(b^\emph{AR}) = 1_{b_k^\emph{AR} > \max \{b^\emph{AR}_{-k}\}}.
					\end{equation}
					\item Pricing rule: The charge to the winning MBP is determined by the second-highest bid for the bidding profile as
					\begin{equation}
						\mathcal{P}^\emph{PViSA}_\emph{vir} : p_k^\emph{AR}(b) =  z_k^\emph{AR}(b^\emph{AR}) \cdot T^\emph{total}_{i,j,k} \max  \{b_{-k}^\emph{AR}\}.
					\end{equation}
				\end{itemize}
			\end{enumerate}
		\end{mechanism}
		
		\subsection{Property Analysis}\label{pvisaproperty}
		We find that the PViSA mechanism is inefficient in allocation. To illustrate the inefficiency, we consider two extreme cases that either (i) sets the reserve at zero (and thus, the brand MBP always loses) or (ii) sets the reserve at an arbitrarily high valuation (and thus, the brand MBP never loses). Therefore, the second-price sealed-bid auction can achieve a surplus of $\max(\gamma\mathbb{E}[X_0], \mathbb{E}[X_{(1)}])$. Suppose that AV $i$ is the winning AV in the physical submarket, we demonstrate how inefficient $\mathcal{M}^\emph{PViSA}$ is by showing its bounds in the following proposition.
		
		\begin{proposition}[Efficiency of PViSA]\label{proposition1} Let $\mathcal{M}_{b^\emph{AR}_0}^\emph{PViSA}$ denote the mechanism of PViSA when a bid of $b_0^\emph{AR}$ is submitted by the brand MBP $0$ (or by the auctioneer on behalf of the brand MBP). There exists $b_0^\emph{AR}\in\{0,\infty\}$ such that for any $\gamma>0, I\geq 2, K\geq 2, F_{i,k}, G_i$, and $\mathbb{E}[m_{\iota, 0}]$, the lower bound is 
			\begin{equation}
				W(\mathcal{M}_{b_0^\emph{AR}}^\emph{PViSA}) \geq \frac{1}{2}W(\mathcal{M}^{*}).
			\end{equation}
			Moreover, there exist $G_i,F_{i,k},K,I,$ and $\mathbb{E}[m_{i,0}]$ for any $\gamma,\epsilon>0$ such that
			\begin{equation}
				\sup_{{b^\emph{AR}_0}\geq 0} W(\mathcal{M}^\emph{PViSA}_{b^\emph{AR}_0}) < (\frac{1}{2}+\epsilon)W(\mathcal{M}^{*}),
			\end{equation}
			where the match qualities $m_{i,k}$ of the performance recommendations and ads are drawn i.i.d. from $F_{i,k}$.
		\end{proposition}
		\begin{proof}
			Please see Appendix \ref{appendix:a}.
		\end{proof}
		
		It is well known that the second-score auction is economically efficient~\cite{krishna2009auction}, such that the allocation in the physical submarket is efficient. The PViSA mechanism is based on two parallel second-price auctions in the physical and the virtual submarkets, respectively. Therefore, the worst case of social welfare obtained by the PViSA mechanism is half of the omniscient mechanism~\cite{arnosti2016adverse}.
		
		\section{The Enhanced Synchronization Mechanism}\label{sec:enhanced}
		In the vehicular Metaverse, the interplay of the physical and virtual submarkets causes externalities and asymmetric information in the synchronization market.
		To address the inefficiency issue of the PViSA mechanism, we apply several advanced techniques in auction theory~\cite{che1993design, arnosti2016adverse} to enhance the auction-based synchronization mechanism by overcoming the externalities in the physical submarket and the asymmetric information in the virtual submarket in Subsection~\ref{epvisa}. For the P2V synchronization, AVs in the physical submarket are allowed to submit their prices and preferred deadlines of DT tasks to the auctioneer. In addition, for the V2P synchronization, the auctioneer can determine the allocation rule according to the received bids from AVs with the synchronization scoring rule. Moreover, for the V2P synchronization, by adopting the price scaling factor $\alpha \geq 1$ in the virtual submarket, the auctioneer can capture a significant fraction of the total surplus from both performance and brand MBPs. Finally, we analyze the properties of the EPViSA mechanism in Subsection~\ref{epvisaproperty}.
		\subsection{Designing the EPViSA Mechanism}\label{epvisa}
		In this subsection, we describe the workflow and property analysis of the enhanced second-score auction-based physical-virtual synchronization mechanism.
		To begin with, the definition of the synchronization scoring rule that is similar to~\cite{che1993design} is provided as follows.

		\begin{definition}[Synchronization Scoring Rule] 
			For any submitted bidding price $q$ and required deadline $\eta$, a synchronization scoring rule $\Phi^\emph{syn}(q, \eta)$ for AVs' bids can be given by
			\begin{equation}
				\Phi^\emph{syn}(q, \eta) = q - \phi(\eta),
			\end{equation}
			where $\phi(\cdot)$ is a non-decreasing function and $\phi(0)=0$.
		\end{definition}
		
		This synchronization scoring rule can be leveraged by the auctioneer to evaluate the synchronizing MBP in the virtual submarket. Based on the synchronization scoring rule, the auctioneer can determine the synchronizing AV in the physical market and eliminate the external effects. The workflow of the EPViSA mechanism can be illustrated as follows.
		
		\begin{mechanism}[EPViSA Mechanism]
			Similar to the PViSA mechanism, the EPViSA mechanism consists of the allocation rules and the pricing rules in the physical and the virtual submarkets, i.e., $\mathcal{M}^\emph{EPViSA} = (\mathcal{Z}^\emph{EPViSA}_\emph{phy}, \mathcal{P}^\emph{EPViSA}_\emph{phy},\mathcal{Z}^\emph{EPViSA}_\emph{vir},\mathcal{P}^\emph{EPViSA}_\emph{vir})$.
			\begin{enumerate}
				\item Physical submarket:
				To allocate synchronization services, AVs submit their multi-dimensional bids $b^\emph{DT} = ((b_1^\emph{DT}, \dots , b_I^\emph{DT}), \eta = (\eta_1, \dots, \eta_I))$ to the auctioneer. The auctioneer computes the scores $\Phi^\emph{syn} = \Phi^\emph{syn}(b^\emph{DT}, \eta) = (\phi^\emph{syn}_1(b_1^\emph{DT}, \eta_1), \dots, \phi^\emph{syn}_I(b_I^\emph{DT}, \eta_I))$ according to the received bids and the synchronization scoring rule. Based on the calculated scores, the auctioneer determines the allocation probabilities and the payments of AVs in the physical submarket.
				\begin{itemize}
					\item Allocation rule: The auctioneer allocates the AV with the highest score as the winning AV, as follows:
					\begin{equation}
						\mathcal{Z}^\emph{EPViSA}_\emph{phy} : z_i^\emph{DT}(\Phi^\emph{syn}) = 1_{\{\Phi^\emph{syn}_i>\max \{\Phi^\emph{syn}_{-i}\}\}}.
					\end{equation}
					\item Pricing rule: Then, the winning AV is charged with the bidding price of the bidder with the second highest score as follows:
					\begin{equation}
						\mathcal{P}^\emph{EPViSA}_\emph{phy} : p_i^\emph{DT}(\Phi^\emph{syn}) = z_i^\emph{DT}(\Phi^\emph{syn}) \cdot b^\emph{DT}_{\arg\max \{\Phi^\emph{syn}_{-i}\}}.
					\end{equation}
				\end{itemize}
				
				\item Virtual submarket: In the virtual submarket, the brand MBP and performance MBPs submit their bids $b^\emph{AR} = (b^\emph{AR}_0, b^\emph{AR}_1, \dots, b^\emph{AR}_K)$ to the auctioneer. Then, the auctioneer determines the allocation probabilities and payments under the scale $\alpha$ of the EPViSA mechanism in the virtual submarket, as follows:
				\begin{itemize}
					\item Allocation rule: For performance MBPs $k = 1,\dots, K$, the allocation probabilities are given by
					\begin{equation}
						\mathcal{Z}^\emph{EPViSA}_\emph{vir} : z^\emph{AR}_{k}(b^\emph{AR})=1_{b^\emph{AR}_{k} > \alpha b^\emph{AR}_{-k}},
					\end{equation}
					for $\alpha \geq 1$.
					Then, the allocation probability for brand MBP $0$ is $z_0^\emph{AR}(b^\emph{AR})\in\{0,1\}$ can be calculate based on the allocation probabilities of performance MBPs $1,\dots,K$, which can be expressed as $z^\emph{AR}_0(b^\emph{AR})\leq 1-\sum_{k=1}^{K}z^\emph{AR}_k(b^\emph{AR})$.
					\item Pricing rule: Finally, the winning MBP is charged with the product of the second highest bid and the price scaling factor $\alpha$ as
					\begin{equation}
						\mathcal{P}^\emph{EPViSA}_\emph{vir}: p^\emph{AR}_k(b^\emph{AR}) = z_k^\emph{AR}(b^\emph{AR})\cdot \rho_k^\emph{AR},
					\end{equation}
					where
					\begin{equation}
						\rho_k^\emph{AR} = \begin{cases}
							T^\emph{total}_{i,j,0}b^\emph{AR}_{0}, &  k=0, \\ 
							T^\emph{total}_{i,j,k}\alpha \max  \{b^\emph{AR}_{-k}\}, & k=1,\dots,K.
						\end{cases}
					\end{equation}
				\end{itemize}
			\end{enumerate}
			
		\end{mechanism}
		Based on the above allocation and pricing rules, we next analyze the economic properties of the proposed mechanism.
		\subsection{Property Analysis}\label{epvisaproperty}
		To maximize its utility, each physical AV $i$ can choose the deadline that can maximize the sum of its valuation $\nu_i$ and externality $\phi(\eta_i)$ for the virtual submarket. In Proposition~\ref{proposition2}, each AV can choose the optimal deadline to maximize its expected payoff.
		\begin{proposition} [Optimal Deadline]\label{proposition2}
			The optimal deadline bidding strategy of the AV $i$ is given by
			\begin{equation}
				\eta_i^* = \arg\max_{\eta\in(0,\eta_i]} (\nu_i + \phi(\eta)).\label{eq:deadline}
			\end{equation}
		\end{proposition}
		\begin{proof}
			Please see Appendix \ref{sec:p2}.
		\end{proof}
		For the optimality of the selection of quality, a similar proof of Proposition \ref{proposition2} can be found in~\cite{tang2018multi}. Based on the bids submitted by AVs and the chosen optimal deadline, the auctioneer can maintain an efficient synchronization scoring rule, which can maximize the social welfare, to guide the allocation decisions in the physical submarket, as follows.
		\begin{definition}[Efficient Synchronization Scoring Rule]\label{definition2} An efficient synchronization scoring rule is in the form of
			\begin{equation}
				\Phi^\emph{syn}(q^\emph{DT}, \eta^*) = q^\emph{DT} + \eta^*[\gamma W_B^\emph{AR}(\mathcal{M}) + W_\emph{P}^\emph{AR}(\mathcal{M})],
			\end{equation}
			where $\eta^* [\gamma W_B^\emph{AR}(\mathcal{M}) + W_\emph{P}^\emph{AR}(\mathcal{M})]$ is the total surplus of MBPs by providing Metaverse billboards.
		\end{definition}
		
		With the above definition of the efficient synchronizing scoring rule, the worst-case analysis of the EPViSA mechanism is given as follows.
		
		\begin{proposition}\label{proposition3}With the efficient synchronization scoring rule in Definition \ref{definition2}, for any $\gamma>0, K, I, F_{i,k}, G_i,$ and $\mathbb{E}[m_{i,0}]$, there exists $\alpha\in\{1,\infty\}$ such that
			\begin{equation}
				W(\mathcal{M}^\emph{EPViSA}_\alpha) \geq \frac{1}{2}W(\mathcal{M}^{*}).
			\end{equation}
		\end{proposition}
		\begin{proof}
			Please see Appendix \ref{sec:p3}.
		\end{proof}

		For a synchronization mechanism, strategy-proofness indicates that participants will not get a higher utility by changing their truthful bids. Adverse selection free indicates that if the existence of market externalities and asymmetric information is independent of bidders' valuations, then under this mechanism, the factors of market externalities and asymmetric information are also independent of the allocation rules of the synchronization mechanisms. As a consequence, it should be highlighted that the EPViSA mechanism is fully strategy-proof and adverse-selection free, as demonstrated in the following theorem.
		\begin{theorem}\label{theorem1}
			The EPViSA mechanism is fully strategy-proof and adverse-selection-free in the synchronization market with the efficient synchronization scoring rule and the cost-per-time model of AR recommendations and ads.
		\end{theorem}
		\begin{proof}
			Please see Appendix \ref{sec:t1}.
		\end{proof}
		Considering the situation where values are drawn independently from a power law distribution, the pessimistic efficiency of the EPViSA mechanism can be guaranteed as follows.
		\begin{theorem}\label{theorem2}
			Given a deterministic K and the efficient synchronization scoring rule, for any $\gamma>0$, there exists $\alpha \geq 1$ such that the following hold simultaneously:
			\begin{equation}
				\begin{aligned}
					W^\emph{DT} (\mathcal{M}^\emph{EPViSA}_\alpha) = W^\emph{DT} (\mathcal{M}^{*}) ,\\
					W_B^\emph{AR} (\mathcal{M}^\emph{EPViSA}_\alpha) = W_B^\emph{AR} (\mathcal{M}^{*}),\\
					W_P^\emph{AR} (\mathcal{M}^\emph{EPViSA}_\alpha ) \geq 0.885\cdot W_P^\emph{AR}(\mathcal{M}^{*}),\\
					W(\mathcal{M}^\emph{EPViSA}_\alpha) \geq 0.96\cdot W(\mathcal{M}^{*}),
				\end{aligned}
			\end{equation}
			when the recommendation match qualities are i.i.d. drawn from a power law distribution.
		\end{theorem}
		\begin{proof}
			Please see Appendix \ref{sec:t2}.
		\end{proof}
		Theorem \ref{theorem2} ensures that there exists a proper selection of $\alpha$,\footnote{When the elements in $H$ are drawn from the power law distribution, the best choice is $\alpha = \max(1,\gamma \mathbb{E}[Q_{\iota,0}]/\mathbb{E}[Q_{\iota,(2)}])$, for the winning AV $\iota$.} which can achieve at least $96\%$ of the total social welfare compared with the omniscient mechanism.
		\section{Experimental and Simulation Results}\label{sec:results}
		
		\begin{table}[!]
			\centering
			\caption{Setting of parameter values}
			\label{tab:parameter}
			
			\begin{tabular}{l||c}
				\toprule
				
				Parameter     & Value \\ \hline
				Number of RSUs $J$ & 1    \\ 
				Number of AVs $I$ & 30    \\ 
				Number of performance MBPs $K$ & 30    \\ 
				Number of brand MBPs         &    1   \\
				Downlink bandwidth $B^d$         &    20 MHz   \\ 
				Uplink bandwidth $B^u$         &    20 MHz   \\ 
				Channel gain $g$         &    $U$[0, 1]   \\ 
				Transmission power $p$ of AV         &    $U$[0, 1] mW   \\ 
				Transmission power  of RSU         &    $U$[0, 10] mW   \\ 
				Additive white Gaussian noise $\sigma$ & $\mathcal{N}(0,1)$\\
				CPU frequency of RSU & 3.6 GHz\\
				GPU frequency of RSU & 19 GHz\\
				Size of DT tasks $s^\emph{DT}$         &    $U$[0, 1] MB   \\ 
				CPU cycles per unit DT data $e^\emph{DT}$         &    $U$[0, 1] Gcycles/MB   \\ 
				Deadline of DT tasks $\eta^\emph{DT}$         &    $U$[0.9, 1.1]  \\
				Size of AR layers $e^\emph{AR}$         &    $U$[0, 1] MB   \\
				GPU cycles per unit AR data $e^\emph{AR}$         &    $U$[0, 1] Gcycles/MB \\ 
				Valuation of DT tasks $\nu$         &    $U$[0, 1]\\
				Recommendation match qualities  $m$ &   $Zipf$(2)\\ 
				Relative bargaining power $\gamma$ & 1\\
				\bottomrule
			\end{tabular}%
		\end{table}
		In this section, we implement the proposed mechanisms and evaluate their performance under different system parameters and various requirements. In addition, we collect real-world eye-tracking data samples by developing a testbed of the vehicular Metaverse and investigating the proposed mechanisms based on the collected data.
		
		\subsection{System Settings}
		In the simulation setup, we consider a synchronization system with parameters listed in Table \ref{tab:parameter}, where $U$ denotes the uniform distribution and $Zipf$ denotes the Zipf distribution. Each simulation is conducted for 10,000 randomly generated system and market scenarios to demonstrate the expected results.
		Following~\cite{arnosti2016adverse}, the price scaling factor is chosen as $\alpha_\iota = \max{(1,\gamma\mathbb[Q_{\iota,0}]/\mathbb[E][Q_{\iota,(2)}])}$, where $\iota$ is the synchronizing AV in the P2V synchronization. The simulation testbed of the vehicular Metaverse is developed based on a 3D model of several city blocks in New York City. The model was constructed by Geopipe, Inc., by using AI to create a digital twin from images taken throughout the city. From there, we simulate an autonomous car driving through a road, with artificially placed highway advertisements on the sides of the road. Human subjects were placed through the simulation, and eye-tracking data was monitored through the HMD Eyes addon produced by Pupil Labs. Afterward, subjects filled out a survey asking them to rate their interest in each ad.
		\begin{figure*}[!]
		\vspace{-0.5cm}
			\centering
			\subfigure[Social welfare v.s. number of AVs.]{\includegraphics[width=0.32\linewidth]{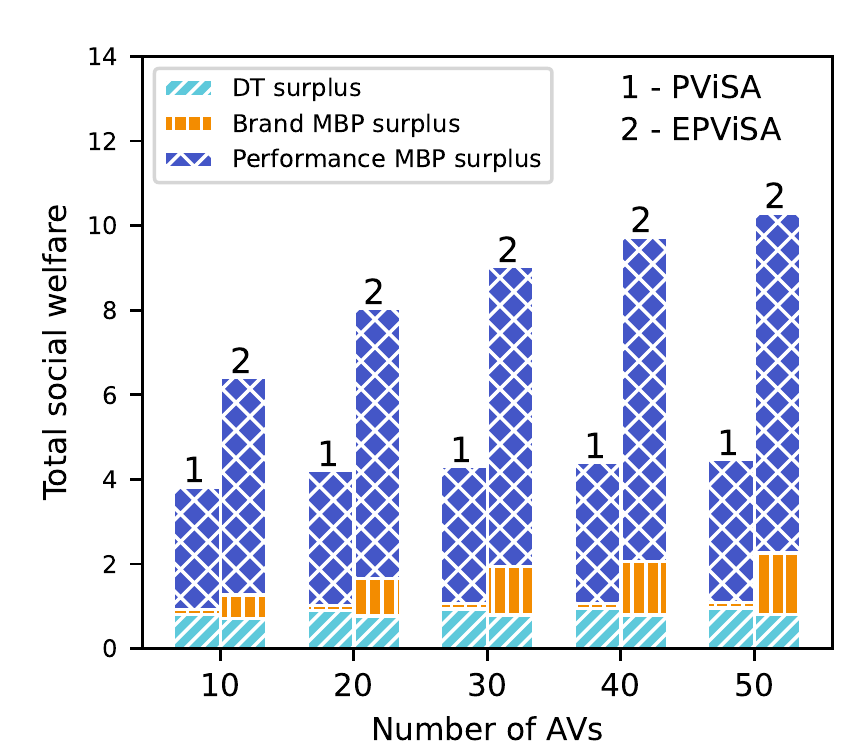}%
				\label{VMAF_SW}}
			\subfigure[Social welfare v.s. number of performance MBPs.]{\includegraphics[width=0.32\linewidth]{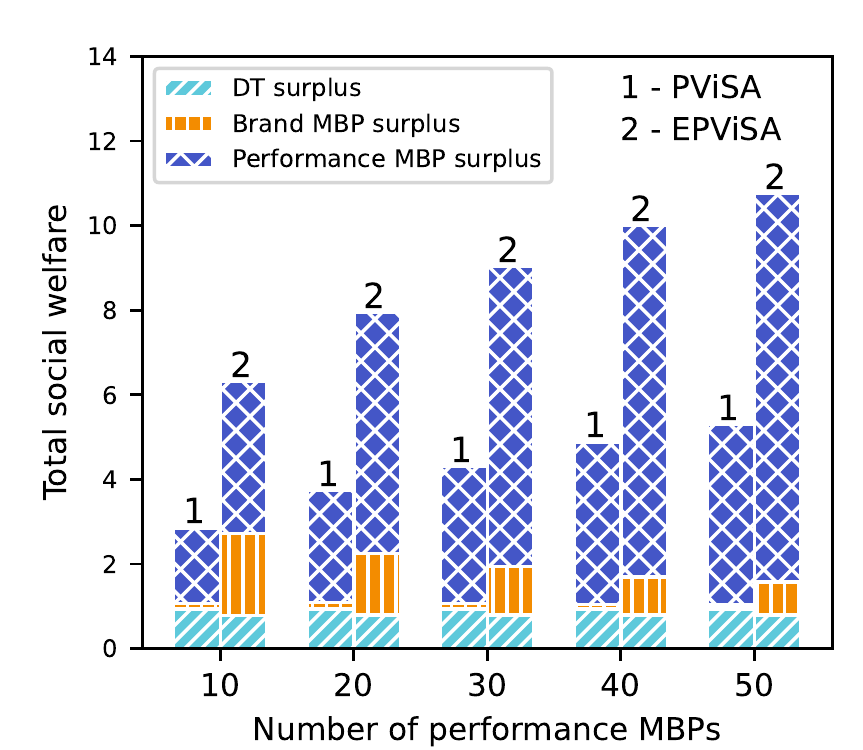}%
				\label{VMAF_BB}}
			\subfigure[Social welfare v.s. size of preference caches.]{\includegraphics[width=0.32\linewidth]{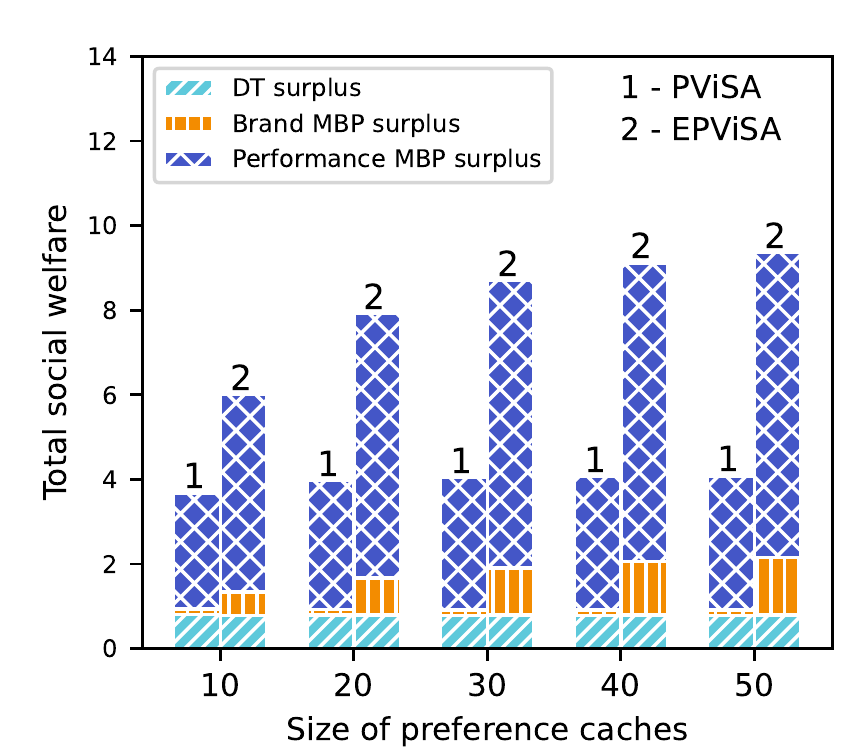}%
				\label{VMAF_BS}}
			\caption{Performance evaluation under different sizes of synchronization market and sizes of preference caches.}
			\label{fig:sw}
		\end{figure*}
		\begin{figure*}[!]
		\vspace{-0.4cm}
			\centering
			\includegraphics[width = 1\linewidth]{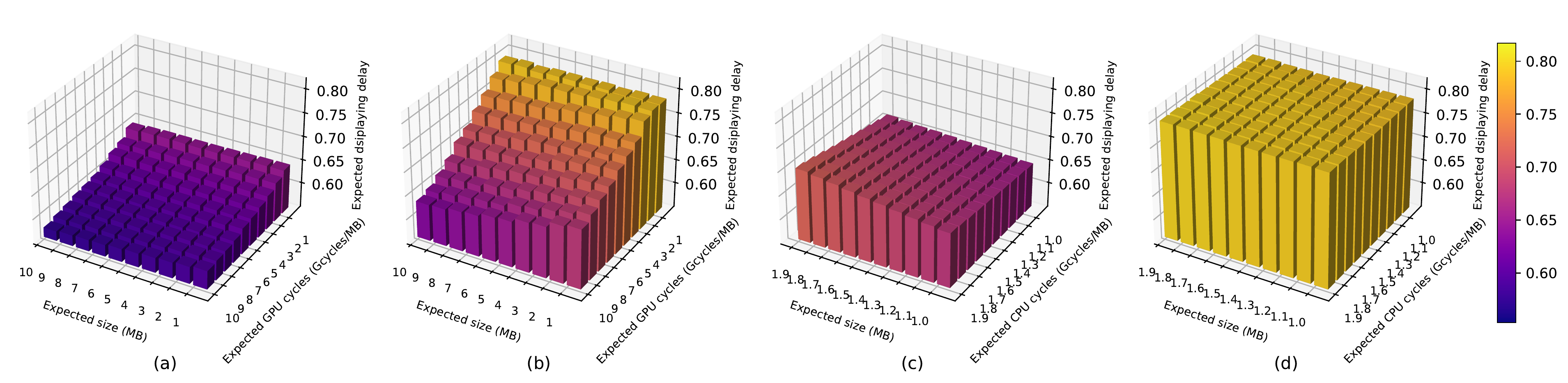}
			\caption{Expected displaying duration under different DT and AR requirements: (a) AR requirements - PViSA; (b) AR requirements - EPViSA; (c) DT requirements - PViSA; (d) DT requirements - EPViSA.}
			\label{fig:3d}
		\end{figure*}
		\subsection{Performance Evaluation under Different Settings}
		In Fig. \ref{fig:sw}, we evaluate the performance of the proposed synchronization mechanisms in synchronization markets with different market parameters, i.e., the number of AVs, the number of MBPs, and the size of preference caches. As we can observe from Fig. \ref{fig:sw}(a) and Fig. \ref{fig:sw}(b), as the market size increases, the total social welfare achieved by both mechanisms increases. The reason is that RSUs can allocate and match the synchronizing AV and MBP pairs from more potential bidders with the high competition when the market size is large. By reducing the asymmetric information among brand MBPs and performance MBPs, the brand surplus achieved by the EPViSA mechanism is always higher than that of the PViSA mechanism under various market settings. However, when the number of AVs increases, the surplus obtained by brand MBPs in the EPViSA mechanism rises and decreases when the number of MBPs increases. As the number of MBPs grows, the competition in the virtual market becomes more intensive, and thus motivates more performance MBPs to submit their bids and win the synchronization services. The massive growth in performance surplus during increased MBPs can also support this claim. As illustrated in Fig. \ref{fig:sw}(c), as the size of the preference caches increases, the overall social welfare of the EPViSA mechanism rises. In contrast, the performance of the PViSA mechanism is not affected by the size of the preference caches. Overall, the EPViSA mechanism can achieve at least two times the total social welfare of the PViSA mechanism in the synchronization market, especially in the surplus of brand MBPs.
		\begin{figure}[!]
			\centering
			\includegraphics[width=1\linewidth]{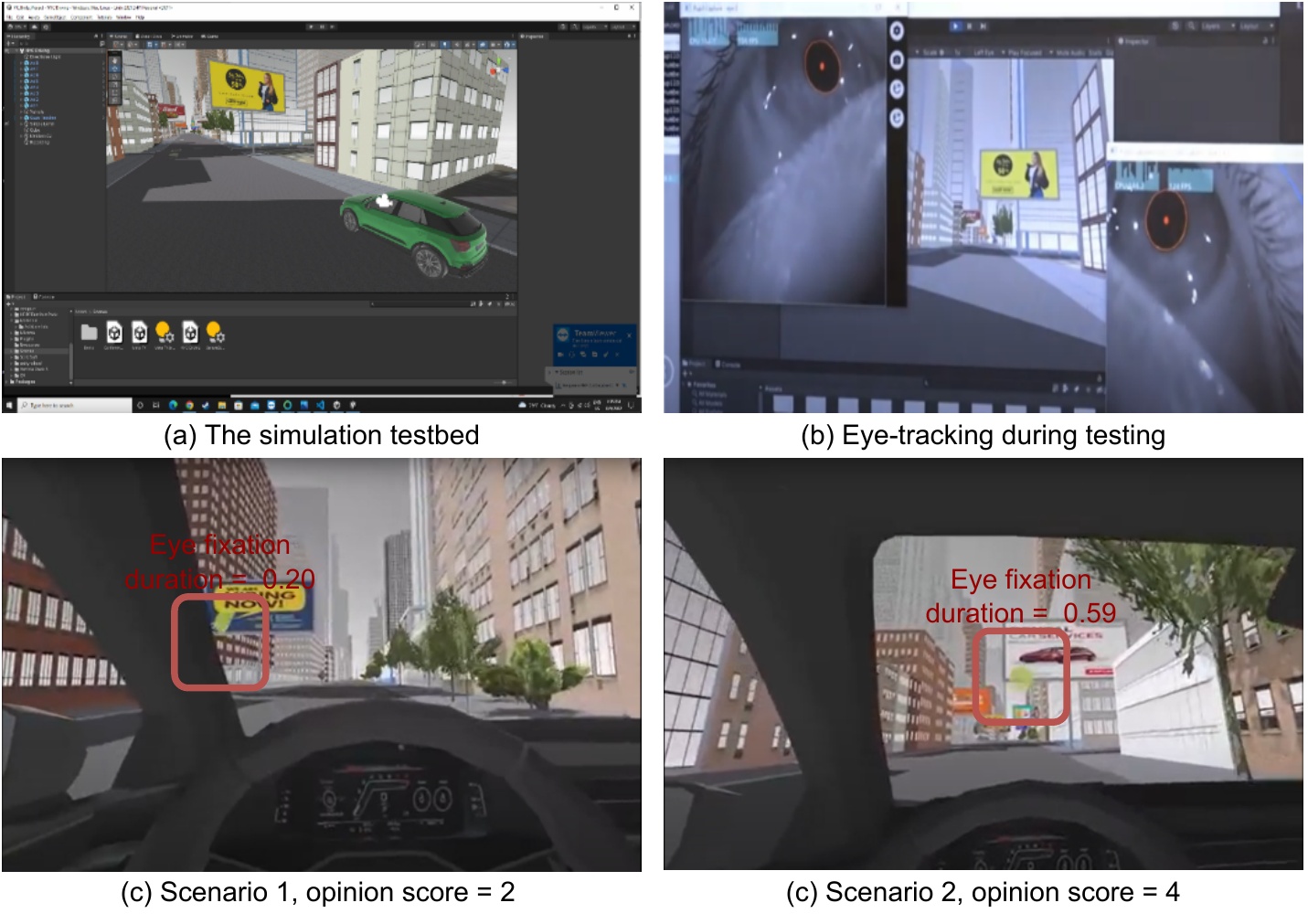}
			\caption{Example screenshots of simulation testbed of the vehicular Metaverse (https://youtu.be/SqhvgRzc5bQ).}
			\label{fig:testbed}
		\end{figure}
		
				\begin{figure*}[t]
				\vspace{-0.5cm}
			\centering
			\subfigure[Social welfare v.s. number of AVs.]{\includegraphics[width=0.32\linewidth]{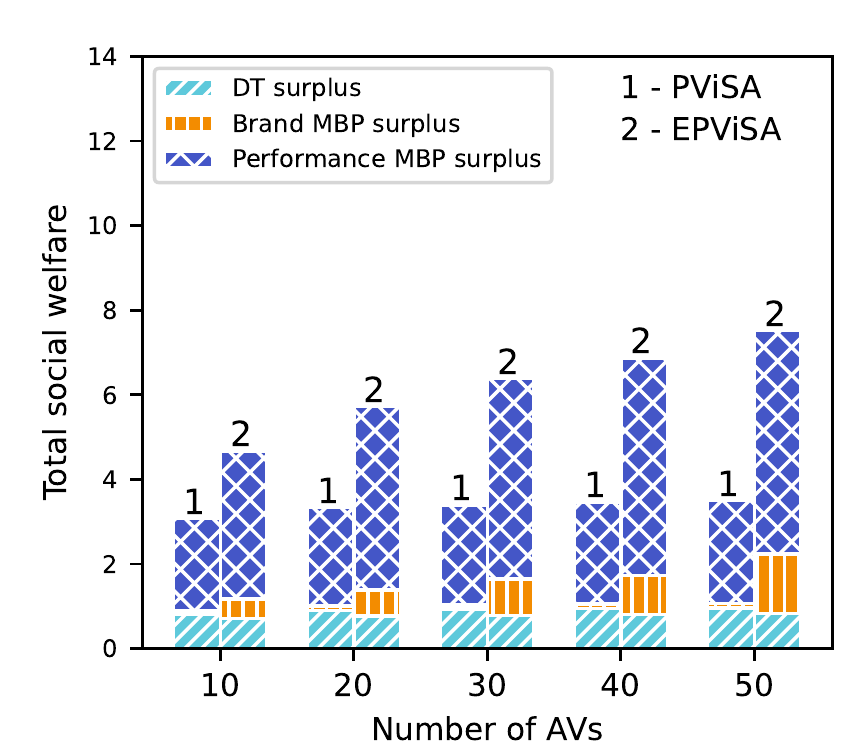}%
				\label{VMAF_SW}}
			\subfigure[Social welfare v.s. number of performance MBPs.]{\includegraphics[width=0.32\linewidth]{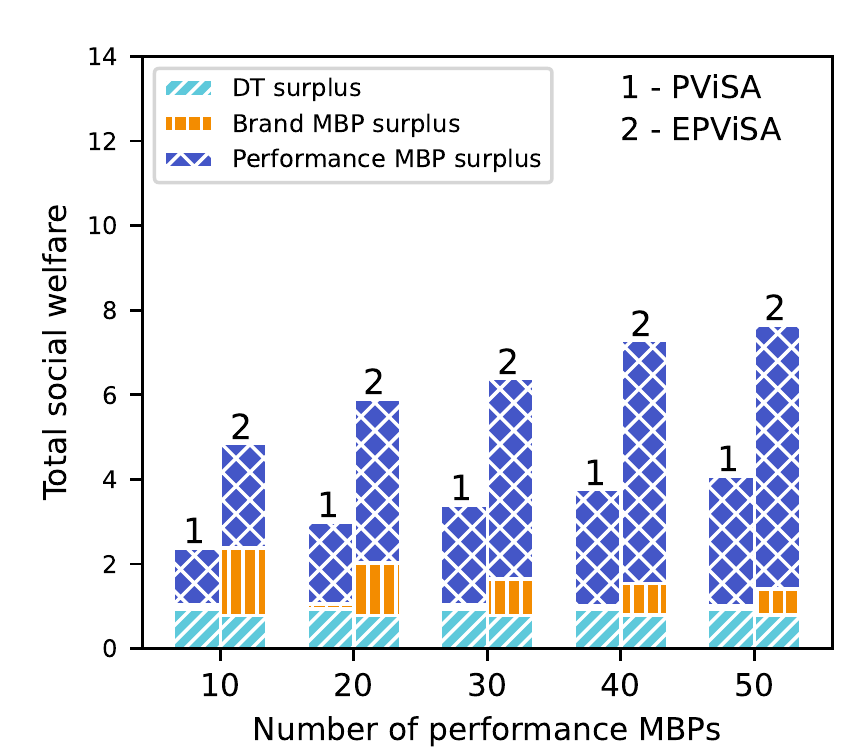}%
				\label{VMAF_BB}}
			\subfigure[Social welfare v.s. size of preference caches.]{\includegraphics[width=0.32\linewidth]{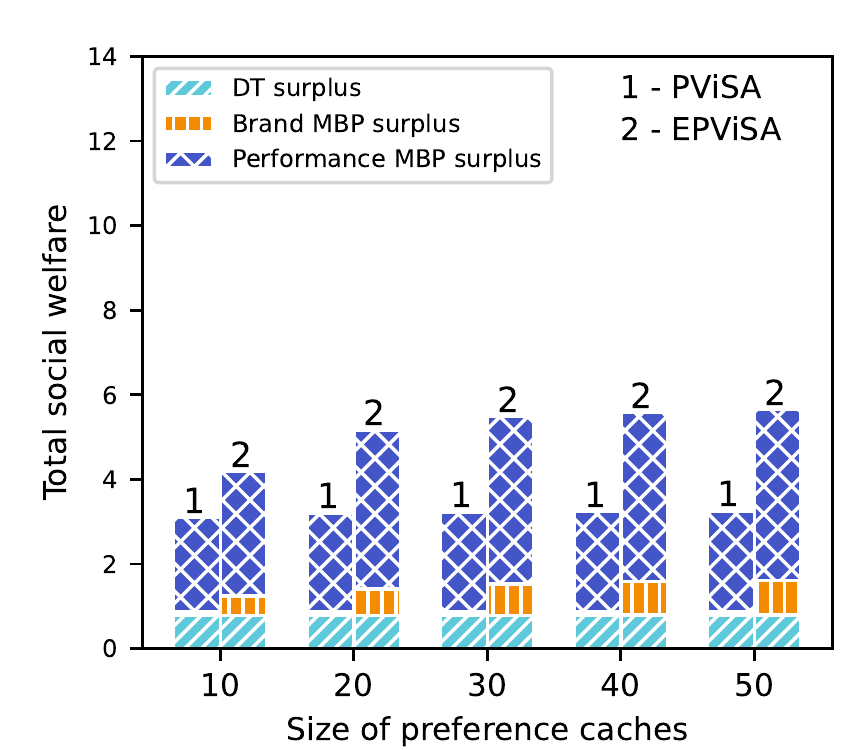}%
				\label{VMAF_BS}}
			\caption{Performance evaluation on real-world data samples under different sizes of synchronization market and sizes of preference caches.}
			\label{fig:real_sw}
		\end{figure*}
		
		\begin{figure*}[!]
		\vspace{-0.4cm}
			\centering
			\includegraphics[width = 1\linewidth]{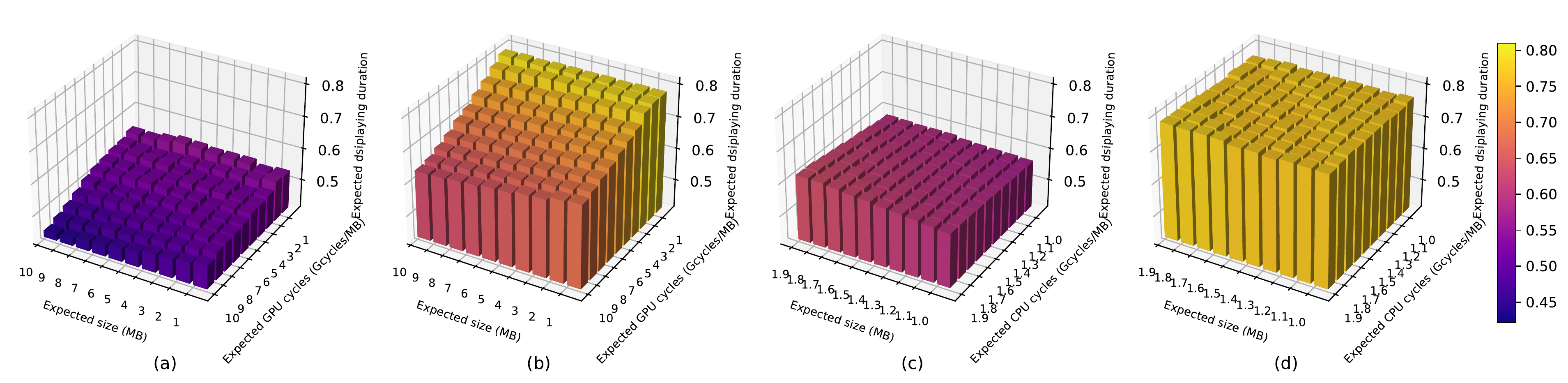}
			\caption{Expected displaying the duration of real-world data samples  under different DT and AR requirements: (a) AR requirements - PViSA; (b) AR requirements - EPViSA; (c) DT requirements - PViSA; (d) DT requirements - EPViSA.}
			\label{fig:real_3d}
		\end{figure*}
		
		\subsection{Expected Displaying Duration under Different DT and AR Requirements}
		To illustrate how the requirements of DT and AR tasks affect the performance of the mechanism, we compare the expected displaying durations achieved by the PViSA and EPViSA mechanisms, as shown in Fig. \ref{fig:3d}. The number of AVs, the number of MBPs, and the number of preference caches are set to 30, 30, and 30, respectively. First, we can observe from Fig. \ref{fig:3d}(a) and Fig. \ref{fig:3d}(b) that the expected displaying duration decreases as the resource requirements of AR tasks become more intensive for both mechanisms. The reason is that as more resources are required for single AR layer rendering, it will be more inefficient for RSUs to provide AR rendering services in the synchronization system. Therefore, RSUs prefer ignoring the AR rendering tasks and focusing on the DT execution services of AVs. To explain further, the efficient synchronization scoring rule can introduce the external information of the virtual submarket during the allocation of synchronizing AV in the physical market, which makes the auctioneer select the AV with the maximum potential surplus. Regarding the physical submarket, the expected displaying durations achieved by the PViSA and EPViSA mechanisms under various DT requirements are shown in Fig. \ref{fig:3d}(c) and Fig. \ref{fig:3d}(d). As we can observe in Fig. \ref{fig:3d}(d), by eliminating the externalities from the virtual submarket, the EPViSA mechanism becomes stable under various DT requirements. In contrast, the PViSA mechanism is affected by the externalities of the virtual submarket. Therefore, its expected displaying duration is lower and more unstable than the EPViSA mechanism. However, the displaying duration in the simulated experiments is only limited by the DT deadline requirements of AVs. In the next part, we will implement a real-world testbed to collect the eye fixation duration of AR users for evaluating the proposed system and mechanism.
		\subsection{Performance Evaluation in the Simulation Testbed}
		To complement the above simulation experiments, we implement a real-world testbed of the vehicular Metaverse in this paper, as shown in Fig.~\ref{fig:testbed}. In this testbed, we use eye-tracking devices to track the Metaverse users' gaze models and record their gaze points and eye fixation duration during testing. In contrast to the ideal experiment, where the display duration does not correspond to the execution delay, the duration of eye fixation that the user captures in the simulation testbed is used to constrain the display duration in the vehicular Metaverse. The opinion score of users given after they finished the testing is leveraged as the input to the power law function of matched preference caches with a base of two. Using these data samples, we retest the performance of the proposed mechanism and plot it in Fig.~\ref{fig:real_sw} and Fig.~\ref{fig:real_3d}. In Fig.~\ref{fig:real_sw}, it can be observed that the total social welfare that can be achieved by the proposed mechanisms is reduced to different degrees. Moreover, the difference in performance between the proposed mechanisms is reduced from three times to twice the original simulation data. In Fig.~\ref{fig:real_3d}, we can see that the trend of implementing the proposed mechanism on different DT and AR requirements does not change significantly compared to the simulation data. However, unlike the idealization of the simulation data, the real data is often more complex. Therefore, the performance of the two mechanisms in Fig.~\ref{fig:real_3d} is not as smooth as that in Fig.~\ref{fig:3d}.

		\section{Conclusions}\label{sec:conclusions}
		In this paper, we have studied the real-time synchronization system and mechanism design problem for the physical-virtual synchronization market of the vehicular Metaverse. We have proposed a real-time synchronization system where physical entities, i.e., AVs, and virtual entities, i.e., MBPs, continuously interact with each other. To address the adverse selection problem, we have proposed the EPViSA mechanism to allocate and match synchronizing pairs of AVs and MBPs. We have provided rigorous analysis and extensive experimental and simulation results to demonstrate the effectiveness and efficiency of the proposed mechanism. Finally, we implement the simulation testbed of the vehicular Metaverse for examining the performance of the proposed mechanisms with real-world data samples.



		\appendix
		\subsection{Proof of Proposition \ref{proposition1}} \label{appendix:a}
		\begin{proof}
			It is well-known that the efficiency of the second-price auction is one with symmetric information. Then, we demonstrate that second-price auctions cannot ensure $(\frac{1}{2} +
			\epsilon)W(\mathcal{M}^*)$, for any $\epsilon > 0$. Fix $I \geq 2$, $K \geq 2$, and $\epsilon > 0$.
			With the synchronizing AV $\iota$, we consider that $m_{\iota,i}$ are i.i.d. drawn from a power law distribution with parameter $a$, which keep the identity $\gamma \mu = \gamma \mathbb{E} [m_{\iota,0}] =
			(1 + \epsilon) E [m_{\iota,(1)}]$. As $a \rightarrow 1^+$, it provides possibility to
			capture nearly all of the surplus from performance MBPs by assigning them a vanishingly small ratio of recommendation positions. Thus, we have
			\begin{equation}
				\begin{aligned}
					\lim_{a \rightarrow 1^+}   \dfrac{W(\mathcal{M}^*)}{\mathbb{E} [\nu]
						\mathbb{E} [m_{\iota,(1)}]} & = \lim_{a \rightarrow 1^+} \frac{\mathbb{E}
						[\max (\gamma \nu, m_{\iota,(1)})]}{\mathbb{E} [m_{\iota,(1)}]}\\
					& \rightarrow \frac{\gamma \mu +\mathbb{E} [m_{\iota,(1)}]}{\mathbb{E}
						[m_{\iota,(1)}]}\\
					& = 2 + \epsilon.
				\end{aligned}
			\end{equation}
			
			The point here is that by choosing $a$ sufficiently close to one, we have
			$W(\mathcal{M}^*) > 2\mathbb{E} [m_{\iota,(1)}] \mathbb{E} [\nu_\iota]$, where $\nu_\iota$ is
			drawn from a power law distribution with parameter $a' .$ If $a'$ is
			sufficiently close to one, then $\sup_{b_0^\emph{AR}} (\mathcal{M}^\emph{PViSA}_{b_0^\emph{AR}}) = \gamma \mu
			\mathbb{E} [\nu\iota] = (1 + \epsilon) \mathbb{E} [m_{\iota,(1)}] \mathbb{E} [\nu_\iota] .$
			It follows that $\sup_{b_0^\emph{AR}} W(\mathcal{M}^\emph{PViSA}_{b_0^\emph{AR}}) / W(\mathcal{M}^*) < \frac{1 +
				\epsilon}{2}$.
		\end{proof}
		\subsection{Proof of Proposition \ref{proposition2}} \label{sec:p2}
		\begin{proof}
			To prove this proposition, we need to follow the idea that for any physical bid
			$(\bar{b}_i^{\tmop{DT}}, \bar{\eta} _i)$, there always exists a physical bid
			$(\hat{b}_i^{\tmop{DT}}, \hat{\eta}_i)$ that results in an expected utility
			of physical bidder $i$ that is no lower than that of the bid $\left(
			\bar{b}_{i \text{}}^{\tmop{DT}}, \bar{\eta} _i \right)$. On the hand,
			the deadline $\hat{\eta}_i$ is calculated according to \eqref{eq:deadline}. On the other
			hand, the bidding price $\hat{b}_i^{\tmop{DT}}$ is selected such that $\Phi^\emph{syn}
			(\hat{b}_i^{\tmop{DT}}, \hat{\eta}_i) = \Phi^\emph{syn} (\bar{b}_i^{\tmop{DT}}, \bar{\eta}
			_i)$, which means that these physical bids $(\bar{b}_i^{\tmop{DT}},
			\bar{\eta} _i)$ and $\left( \bar{b}_{i \text{}}^{\tmop{DT}}, \bar{\eta} _i
			\right)$ result in the same score as well as the same allocation probability.
			In the case of the bidder loses, the utility with be zero if bidder $i$ loses by
			obtaining this score. In the other case of bidder wins, and the bid
			$(\hat{b}_i^{\tmop{DT}}, \hat{\eta}_i)$ introduces higher or equal utility
			that obtained by submitting other bids, i.e.,
			\begin{equation}
				\begin{aligned}
					\nu_i - \hat{b}_i^{\tmop{DT}} - &(\max \{ \Phi^\emph{syn}_{\mathcal{I} / \{ i \}} \} + \phi
					(\hat{\eta} _i)) \geq \\&\nu_i - \bar{b}_i^{\tmop{DT}} - (\max \{
					\Phi^\emph{syn}_{\mathcal{I} / \{ i \}} \} + \phi (\bar{\eta} _i)) .
				\end{aligned}
			\end{equation}
			This holds as the deadline is selected by calculating \eqref{eq:deadline}.
		\end{proof}
		\subsection{Proof of Proposition \ref{proposition3}} \label{sec:p3}
		\begin{proof}
			Under the efficient synchronization scoring rule, the allocated synchronizing AV in the physical submarket can maximize the expected surplus in the synchronization market. Then, consider that $\alpha = 1$ and $\alpha = \infty$ indicate the cases where the
			brand MBP always loses or never loses, respectively. Therefore, we have $W
			(\mathcal{M}_1^{\tmop{EPViSA}}) =\mathbb{E} [Q_{\iota, (1)}]$ and $W
			(\mathcal{M}_{\infty}^{\tmop{EPViSA}}) = \gamma \mathbb{E} [Q_{i, 0}]$.
			Furthermore, it follows that
			\begin{equation}
				\begin{aligned}
					W (\mathcal{M^{*}})  & = \mathbb{E} [\max (\gamma \nu
					\mathbb{E} [m_{\iota,0}]), \nu m_{\iota,(1)}]\\
					& \leq \gamma \mathbb{E} [Q_{\iota,0}] + \mathbb{E} [Q_{\iota,(1)}]\\
					& \leq 2 \cdot \max (\gamma \mathbb{E} [Q_{\iota,0}], \mathbb{E} [Q_{\iota,(1)}]),
				\end{aligned}
			\end{equation}
			which proves this claim.
		\end{proof}
		\subsection{Proof of Theorem \ref{theorem1}} \label{sec:t1}
		\begin{proof}
			Any mechanism is strategy-proof if and only if this
			the mechanism can be characterized by a critical payment function $\psi$, such
			that bidder $n$ wins if and only if its bid $b_n$ is higher than the
			threshold price $\psi (b_{- n})$ under other competing bids $b_{- n}$. After bidder $n$ wins the auction, the payment charged by the auctioneer
			is the required payment $\psi$. To prove that the EPViSA mechanism is fully
			strategy-proof in the synchronization market, we need to find the
			critical payment functions of the EPViSA mechanism for the physical
			and the virtual submarket, respectively. First, we prove the payment
			$\psi_{\tmop{phy}} (b^\emph{DT}_{- i})$ of EPViSA is critical in the physical
			submarket. Given any deadline $\eta_i$, a bidder $i$ in the physical submarket
			will be scored by $\Phi^{\tmop{syn}} (\nu_i) - \phi (\eta_i)$ if it submits
			a truthful bid. We need to show that bidder $i$ cannot
			manipulate its bid to obtain a higher utility. In the case of the bidder
			losing by submitting an untruthful bid $b_i' \neq b^\emph{DT}_i$, the utility of bidder
			$i$ will be zero. On the contrary, when bidder $i$ wins by submitting an
			untruthful bid, the expected payoff can be expressed as
			\begin{equation}
				\begin{aligned}
					U^{\tmop{DT}}_i & = \nu_i - p_i'\\
					& = \nu_i - (\max \{ \Phi^{\tmop{syn}}_{- i} \} + \phi (\eta_i))\\
					& = \Phi^{\tmop{syn}}_i - \max \{ \Phi^{\tmop{syn}}_{- i} \},
				\end{aligned}
			\end{equation}
			where $\max \{ \Phi^{\tmop{syn}}_{- i} \}$ is the highest score without
			considering the bid of bidder $i$. Therefore, whether bidder $i$ wins or loses
			under both $\Phi'$ and $\Phi^m$, the utilities it perceived are all lower
			or equal to the utility when it submits its truthful bid. The
			critical payment function in the physical submarket can be expressed as
			$\psi_{\tmop{phy}} (b^{\tmop{DT}}_i, \eta_i ) = \max \{ \Phi^{\tmop{syn}}_{-
				i} \} + \phi (\eta_i).$ Moreover, the auctioneer needs to calculate the
			synchronization scores for bidders in the physical bidder. Therefore, all
			bidders are proof of false-name biding.
			
			Then, the critical payment function of the EPViSA mechanism in the virtual
			submarket is $\psi_{\tmop{vir}}  (b^{\tmop{AR}}_{- k}) = \alpha \max \{
			b^{\tmop{AR}}_{- k} \}$, for $\alpha \geq 1$. With the critical payment
			function in the virtual submarket, the top performance bidder can win by demands
			that $\psi_{\tmop{vir}}  (b^{\tmop{AR}}_{- k}) \geq \max \{ b^{\tmop{AR}}_{-
				k} \}$. Moreover, the mechanism is false-name proof in the virtual
			submarket, if $\psi_{\tmop{vir}}  (b^{\tmop{AR}}_{- k}) = \psi_{\tmop{vir}} 
			(\max \{ b^{\tmop{AR}}_{- k} \})$. Let there be a set of bids $b_{- k}$ that
			makes $\psi_{\tmop{vir}}  (b^{\tmop{AR}}_{- k}) \neq \psi_{\tmop{vir}} 
			(\max \{ b^{\tmop{AR}}_{- k} \})$, and there are two bidders in the virtual
			submarket. One bidder with valuation higher than $\psi_{\tmop{vir}}  (b_{- k})$
			and the other with valuation $\max \{ b^{\tmop{AR}}_{- k} \} .$ If
			$\psi_{\tmop{vir}}  (b^{\tmop{AR}}_{- k}) < \psi_{\tmop{vir}}  (\max \{
			b^{\tmop{AR}}_{- k} \})$, then the first bidder can submit a lower price
			while maintaining the other bids in the set of bids $b_{- k}$. This way,
			the mechanism is not winner false-name proof. On the contrary, if
			$\psi_{\tmop{vir}}  (b^{\tmop{AR}}_{- k}) > \psi_{\tmop{vir}}  (\max \{
			b^{\tmop{AR}}_{- k} \})$ then the losing bidder in the virtual submarket can
			submit a higher bid compared with the winner's bid while maintaining the
			other bids in the set of bids $b^{\tmop{AR}}_{- k}$. Therefore, the
			the mechanism in the virtual submarket is loser false-name proof.
			
			To demonstrate that our mechanism is adverse selection free, the physical
			critical payment function is quasi-linear, and the virtual critical payment
			the function is homogeneous of degree one. In the physical submarket, we consider
			two types of external effects from the virtual submarket $\eta \in \{ 0, \infty
			\}$ with $\Pr (\phi = 0) \in (0, 1)$, while the set of other bidding price
			$\nu_{- i}$ maintain the same. When $\phi = 0$, there is no external
			effects from the virtual submarket, and we have $z_{i }^{\tmop{DT}} (\nu + \phi (0) ) =
			z_{i }^{\tmop{DT}} (\nu) = 1_{\nu_i > \max \nu_{- i}} = 1_{\nu_i >
				\psi_{\tmop{phy}} (\nu_{- i}, 0)}$. When $\phi = \sigma, z_{i }^{\tmop{DT}}
			(\nu + \phi (0) ) = z^{\tmop{DT}}_{i } (\nu + \infty) = 1_{\nu_i >
				\psi_{\tmop{phy}} (\nu_{- i}, \phi(\infty))} = 0 $, there is no bidder can
			win in the physical submarket. Therefore, the proposed mechanism in the
			the physical submarket is adverse selection free. When it turns to the virtual
			submarket, suppose that $\nu \in \{ 1, c \}$ with $\Pr (C = 1) \in (0, 1)$
			while maintaining the same set of Metaverse billboard qualities. It can
			be shown that $z^{\tmop{AR}}_{0 } (\nu m) = 1_{\{ \nu = c \}} .$ When $\nu =
			1$, $z^\emph{AR}_i (\nu m) = z^\emph{AR}_i (\nu m) = 1_{\{ m_i > \psi_{\tmop{vir}} (m_{- i}) \}}
			= 1$, and thus $z^{\tmop{AR}}_0 (\tmop{cm}) = 0.$ When $\nu = c, z^\emph{AR}_i (\tmop{cm})
			= 1_{\{ c m_i > \phi_{\tmop{vir}} (c m_{- i}) \}} = 0.$, this implies no \
			performance bidder can win the auction and then $z^\emph{AR}_0 (\tmop{cm}) = 1$.
			
			In summary, we prove that the EPViSA mechanism is fully strategy-proof for bidders in the synchronization market and adverse selection free under the efficient synchronization scoring rule and cost-per-time model.
		\end{proof}
		\subsection{Proof of Theorem \ref{theorem2}} \label{sec:t2}
		\begin{proof}
			We first prove that under the optimal deadline selection in Proposition \ref{proposition2} and the truthful bid report, the EPViSA mechanism is efficient under the efficient synchronization scoring rule, and the synchronizing AV is $\iota$.
			To prove the efficiency of EPViSA in the virtual submarket, we first set the
			common valuation and the expected displaying duration to be identically one, independent of the allocation rule in the virtual market. Moreover, the efficiency of the allocation in
			the virtual submarket is influenced by the following three parameters: (i) the
			number of performance bidders $K$, (ii) the expected match quality of brand MBP $\mu
			=\mathbb{E} [m_{\iota, 0}]$, and (iii) the parameters of the power law tail, $a$. We denote
			that $\lambda$ to be the probability that the brand MBP wins the recommendation
			position under omniscient. Let $\mu (\lambda, K)$ refer to the brand quality
			implied by the provided qualities of $\lambda$ and $K$ (for static $a$), and
			thus $\lambda = \Pr (m_{\iota,(1)} \leq \mu (\lambda, K))$. 
			The total social welfare of the omniscient benchmark can be calculated as
			\begin{equation}
				\begin{aligned}
					W(\mathcal{M}^*) & = 1 + W_B (\mathcal{M}^*) + W_P (\mathcal{M}^*)\\
					& = 1+ \lambda \mu (\lambda, K) + \int^1_{\lambda} \mu (x, K) \tmop{dx} . 
				\end{aligned}
			\end{equation}
			The social welfare obtained by the EPViSA mechanism in the virtual submarket can
			be calculated as
			\begin{equation}
				\begin{aligned}
					W_P (\mathcal{M}^\emph{EPViSA}_{\alpha}) & = \mathbb{E} \left[ Q_{\iota, (2)}  \frac{m_{\iota,(1)}}{m_{\iota,(2)}}
					1_{\frac{m_{\iota,(1)}}{m_{\iota,(2)}}>\alpha}\right]\\
					& = \mathbb{E} [Q_{\iota,(2)}] \mathbb{E} \left[ \frac{m_{\iota,(1)}}{m_{\iota,(2)}} 1_{\frac{m_{\iota,(1)}}{m_{\iota,(2)}}>\alpha}\right]\\
					& = \mathbb{E} [Q_{\iota,(2)}] \mathbb{E} \left[ \frac{m_{\iota,(1)}}{m_{\iota,(2)}} \right] \alpha \textrm{Pr} \left(
					\frac{m_{\iota,(1)}}{m_{\iota,(2)}} > \alpha \right)\\
					& = \mathbb{E} \left[ Q_{\iota,(2)}  \frac{m_{\iota,(1)}}{m_{\iota,(2)}} \right] \alpha \textrm{Pr} \left(
					\frac{m_{\iota,(1)}}{m_{\iota,(2)}} > \alpha \right)\\
					& = \mathbb{E} [Q_{\iota,(1)}] \alpha^{1 - a}.
				\end{aligned}
			\end{equation}
			Therefore, for any $\alpha \geq 1$ and $W_P (\mathcal{M}^\emph{EPViSA}_{\alpha}) =
			\alpha^{1 - a} \mathbb{E} [Q_{\iota, (1)}]$, we have
			\begin{equation}
				\begin{aligned}
					W &(\mathcal{M}^\emph{EPViSA}_{\alpha}) = 1+W_B (\mathcal{M}^\emph{EPViSA}_{\alpha}) + W_P
					(\mathcal{M}^\emph{EPViSA}_{\alpha})\\
					& = 1+ \Pr (m_{\iota,(1)} \leq \alpha m_{\iota,(2)}) \mu (\lambda, K) + W_P
					(\mathcal{M}^\emph{EPViSA}_{\alpha})\\
					& = 1+(1 - \alpha^{- a}) \mu (\lambda, n) + \alpha^{1 - a} \mathbb{E}
					[m_{\iota,(1)}],
				\end{aligned}
			\end{equation}
			where the last line is obtained from the independence of any power law distribution. 
			
			We select the parameter $\alpha$ and then the brand
			MBP is allocated for the synchronization service with probability ${\lambda}$. In specific, we select ${\alpha}$ that ensures $1-{\alpha}^{-a}={\lambda}$. As both allocation rules deliver a representative sample of recommendation positions to the
			brand MBP, the first claim in Theorem 2 follows directly. In detail, our selection of ${\alpha}$ ensures that
			$W_\emph{B}(\mathcal{M}^*)=W_\emph{B}(\mathcal{M}^\emph{EPViSA}_{\alpha})$.
			
			Taking $\lambda \rightarrow 1$ and applying the L'Hosptial's rule, we can obtain
			that
			\begin{equation}
				\begin{aligned}
					\frac{W_P(\mathcal{M}^\emph{EPViSA}_\alpha)}{W_P(\mathcal{M}^{*})} &= \lim_{\lambda \rightarrow 1}   \frac{(1 - \lambda)^{1 - 1 / a} \Gamma (1 -
						1 / a)}{\int_{\lambda}^1 \log (1 / x)^{- 1 / a} \tmop{dx}} \\& = (1 - 1 / a)
					\Gamma (1 - 1 / a) \lim_{\lambda \rightarrow 1} \frac{(1 - \lambda)^{- 1 /
							a}}{\log (1 / \lambda)^{- 1 / a}}\\
					& = \Gamma (2 - 1 / a) .
				\end{aligned}
			\end{equation}
			By employing the fact of Gamma function $\Gamma(\cdot)$, the last line follows from the identity
			$\Gamma (s + 1) = s \Gamma (s)$ and the fact that $\lim_{\lambda \rightarrow
				1} (1 - \lambda) / \log (1 / \lambda) = 1.$ Because $a > 1,$we have $2 - 1 / a
			\in (1, 2) .$ The lowest value of the Gamma function within $(1, 2)$
			is higher than 0.885, finishing the proof of the third statement in Theorem \ref{theorem2}.
			We can conclude that
			\begin{equation}
				\begin{aligned}
					\lim&_{K \rightarrow \infty} K^{- 1 / a} W (\mathcal{M}^\emph{EPViSA}_{\alpha} ; K, \mu
					(\lambda, K)) \\&= 1 + \lambda \log (1 / \lambda)^{- 1 / a} + \Gamma (1 - 1 /
					a)  (1 - \lambda)^{1 - 1 / a} .
				\end{aligned}
			\end{equation}
			\begin{equation}
				\begin{aligned}
					\lim&_{K \rightarrow \infty} K^{- 1 / a} W (\mathcal{M}^* ; K, \mu (\lambda,
					K)) \\& = 1 + \lambda \log (1 / \lambda)^{- 1 / a} + \int_{\lambda}^1 \log (1 /
					x)^{- 1 / a} dx.
				\end{aligned}
			\end{equation}
			As a result, the proportion of these equations is a lower bound on $W(\mathcal{M}^\emph{EPViSA}_\alpha) =
			W(\mathcal{M}^{*})$. The lowest value of this proportion for $\lambda \in (0, 1)$
			and $1 / a \in (0, 1)$ is higher than 0.960 and thus finishing the proof of the last statement.
		\end{proof}

		
		
		%
		\bibliographystyle{IEEEtran}
		\bibliography{bare_conf}

	\end{document}